\documentclass[a4paper,12pt]{amsart}

\usepackage{amsmath,amssymb,amsthm}
\usepackage{graphicx,stmaryrd}
\usepackage{hyperref}

\newcommand{\CC}{\mathbb{C}}
\newcommand{\EE}{\mathbb{E}} 
\newcommand{\FF}{\mathbb{F}}

\newcommand{\KK}{\mathbb{K}}

\newcommand{\PP}{\mathbb{P}} 
\newcommand{\RR}{\mathbb{R}}

\newcommand{\ZZ}{\mathbb{Z}}

\newcommand{\cA}{\mathcal{A}}

\newcommand{\cC}{\mathcal{C}}
 
\newcommand{\cF}{\mathcal{F}}
\newcommand{\cE}{\mathcal{E}}
\newcommand{\cG}{\mathcal{G}}

\newcommand{\cN}{\mathcal{N}}

\newcommand{\cX}{\mathcal{X}}

\renewcommand{\a}{\alpha}
\newcommand{\D}{\Delta} 
\renewcommand{\d}{\delta} 
\newcommand{\G}{\Gamma}

\newcommand{\g}{\gamma} 
\renewcommand{\L}{\Lambda}

\renewcommand{\l}{\lambda}
\renewcommand{\b}{\beta} 
\renewcommand{\k}{\kappa} 
\newcommand{\Om}{\Omega}
 
\renewcommand{\S}{\Sigma} 
\newcommand{\s}{\sigma}

\newcommand{\eps}{\varepsilon}

\newcommand{\el}{\langle} 
\newcommand{\er}{\rangle}
\newcommand{\tr}{\mathrm{tr}}

\newcommand{\sd}{\triangle}

\newcommand{\fk}{{\sc fk}}

\newcommand{\tfim}{{\sc tfim}}
\newcommand{\lra}{\leftrightarrow}

\renewcommand{\b}{\beta}
\newcommand{\oo}{\infty}

\newcommand{\sm}{\setminus}

\newcommand{\es}{\varnothing}
\newcommand{\se}{\subseteq}

\newcommand{\ol}{\overline}

\newcommand{\crit}{\mathrm{c}}

\newcommand{\one}{\hbox{\rm 1\kern-.27em I}}

\newcommand{\p}{\mathfrak{p}}
\newcommand{\f}{\mathfrak{f}}
\newcommand{\w}{\mathfrak{w}}
\newcommand{\fs}{\mathfrak{s}}
\newcommand{\ft}{\mathfrak{t}}

\newtheoremstyle{slthm}% theoremstyle with slanted body type
     {}%      Space above (empty=\topspace)
     {\baselineskip}%      Space below 
     {\slshape}%         Body font
     {\parindent}%    Indent amount (\parindent = para indent)
     {\scshape}% Thm head font
     {.}%        Punctuation after thm head
     { }%     Space after thm head: " " = normal interword space
     {}%         Thm head spec (can be left empty, meaning `normal')

\theoremstyle{slthm}
\newtheorem{definition}{Definition}[section]
\newtheorem{theorem}[definition]{Theorem}
\newtheorem{proposition}[definition]{Proposition}
\newtheorem{lemma}[definition]{Lemma}

\newtheorem{remark}[definition]{Remark}

\allowdisplaybreaks

\title[Vanishing critical magnetization]
{Vanishing critical magnetization\\
in the quantum Ising model}
\author{Jakob E. Bj\"ornberg}
\thanks{Department of Mathematics,
Uppsala University, Box 256, 751 05 Uppsala, Sweden, 
Phone +46(0)18-471 3106,
e-mail: jakob@math.uu.se}
\date{\today}

\begin{document}

\begin{abstract}
Adapting the recent argument of
Aizenman, Duminil-Copin and Sidoravicius
for the classical Ising model,  it is shown here
that the  magnetization in the
transverse-field Ising model vanishes at the critical
point.  The proof applies to the ground state
in dimension $d\geq2$ and to positive-temperature
states in dimension $d\geq 3$, and relies on graphical
representations as well as an infrared bound.
\end{abstract}

\maketitle

\section{Introduction}\label{intro_sec}

This article concerns the 
transverse-field Ising model,  introduced 
in~\cite{lsm}
and a well-known generalization of
the familiar (classical) Ising model for ferromagnetism.
The model  possesses
a phase transition and  a critical
point, which may be identified using the 
\emph{(residual) magnetization}.
The magnetization equals zero
below the critical point and is positive above it.
An important result for the classical
Ising model is that
the magnetization also vanishes
\emph{at} the critical point:  in two
dimensions this goes back to the work of 
Onsager~\cite{onsager}, in dimension $d\geq4$ it 
was first proved by Aizenman and 
Fern\'andez~\cite{af}, and recently the final
case $d=3$ was established by Aizenman, Duminil-Copin 
and Sidoravicius~\cite{adcs}.
Building on the methods of~\cite{adcs}, the present 
work shows that the magnetization
in the transverse-field model also vanishes
at the critical point.  This implies that there
is a unique equilibrium state at the critical point.
We give precise statements
shortly, but first introduce the relevant 
notation and definitions.

Let $n\geq1$ and write
\[
\L=\L_n=[-n,n]^d=\{-n,-n+2,\dotsc,n-1,n\}^d
\]
for a finite box in $\ZZ^d$.  The transverse-field
Ising model is defined via its Hamiltonian, which in the
finite volume $\L$ takes the form
\[
H_\L=-\l\sum_{xy\in\L} \s_x^{(3)}\s_y^{(3)}-\d\sum_{x\in\L} \s_x^{(1)}
-\g \sum_{x\in\L} \s_x^{(3)}.
\]
Here the first sum is over all (unordered) 
nearest neighbours in $\L$,
\[
\s^{(1)}=\begin{pmatrix}
1 & 0\\
0 & -1
\end{pmatrix},
\qquad
\s^{(3)}=\begin{pmatrix}
0 & 1\\
1 & 0
\end{pmatrix}
\]
are the spin-$\tfrac{1}{2}$ Pauli matrices,
and $\s^{(i)}_x=\s^{(i)}\otimes\mathrm{Id}_{\L\sm\{x\}}$.
The parameters $\l$ and $\d,\g$
are nonnegative and represent spin-coupling
and field-strengths, respectively.
$H_\L$ is an operator on the Hilbert space
$\bigotimes_{x\in\L}\CC^2$, and one defines 
for each $\b\in(0,\oo)$ the state
$\langle\cdot\rangle_{\L,\b}$ by
\begin{equation}\label{Q-state}
\langle Q\rangle_{\L,\b}=\frac{\tr(Qe^{-\b H_\L})}{\tr(e^{-\b H_\L})}.
\end{equation}
The parameter $\b$ is referred to as \emph{inverse temperature}. 
(Readers with a probabilistic background may prefer the `path
integral' definition of $\langle\cdot\rangle_{\L,\b}$
given in Section~\ref{spin_sec}.)

Of particular interest are the one- and two-point functions
\[
\langle \s_x^{(3)}\rangle_{\L,\b}\quad\mbox{ and }\quad
\langle \s_x^{(3)}\s_y^{(3)}\rangle_{\L,\b},
\]
and more general correlation functions
\begin{equation}\label{corr_eq}
\langle \s_A^{(3)}\rangle_{\L,\b},\quad\mbox{ where }
\s_A^{(3)}=\prod_{x\in A}\s_x^{(3)}.
\end{equation}
As written, these are analytical functions
of the model parameters $\l,\d,\g$.  However,
one is interested in their their
limits as $n\to\oo$ or $\b,n\to\oo$.
(Existence of the limits is well-known, 
see eg~\cite{akn} or~\cite{bjo_phd}.)
These  need not be analytical, or even continuous:  
for example if $\g=0$ then by symmetry
$\langle \s_0^{(3)}\rangle_{\L,\b}=0$, whereas the residual
magnetization 
\begin{equation}\label{resid_eq}
\begin{split}
M^+_\b(\l,\d)&:=\lim_{\g\downarrow0}\lim_{n\to\oo}
\langle \s_0^{(3)}\rangle_{\L,\b}\\\mbox{or }
M^+_\oo(\l,\d)&:=\lim_{\g\downarrow0}\lim_{n\to\oo}\lim_{\b\to\oo}
\langle \s_0^{(3)}\rangle_{\L,\b}
\end{split}
\end{equation}
may be strictly positive.  This leads to the definition of the
\emph{critical point} 
\[
\l_\crit=\l_\crit(\d,\b):=\inf\{\l\geq 0: M^+_\b(\l,\d)>0\}.
\]
Note that $\l_\crit$ may also be defined in terms of 
the uniqueness of the infinite-volume states (for all
boundary conditions), or the divergence of the 
susceptibility, see~\cite{bjogr}.
We have that $0<\l_\crit<\oo$ if $d\geq2$, or if
$\b=\oo$ and $d\geq 1$.
The case $\b=\oo$ is referred to as the \emph{ground state}
and the case $\b<\oo$ as \emph{positive temperature}.

The following is the main result of this work:
\begin{theorem}\label{main_thm}
Let $\d>0$.
If $\b=\oo$ and $d\geq2$, or $\b<\oo$ and $d\geq3$, then
the residual magnetization satisfies
\[
M^+_\b(\l_\crit,\d)=0.
\]
\end{theorem}
If $\d=0$ one recovers the
classical Ising model;  it is then standard to take $\l=1$
and vary the parameter $\b$, giving the critical point $\b_\crit$
which is positive and finite if $d\geq2$.  
As remarked above, in this case the result is well-known.  
For $\d>0$ the case when $\b=\oo$ and $d=1$ 
is also known and was
established in~\cite{pfeuty}
(and reproved in~\cite{bjogr} using graphical methods).
The other cases are new.  Note that the
only nontrivial case left open by Theorem~\ref{main_thm} is when
$\b<\oo$ and $d=2$, which remains open  for $\d>0$.

Like the previous works~\cite{adcs,af}
on the classical model, our proof of Theorem~\ref{main_thm} uses
\emph{graphical representations}.  For the classical Ising model, and
related models such as the Potts model, the use of graphical
representations is a standard tool and has been a huge success
since the seminal work of 
Fortuin and Kasteleyn~\cite{fk}.  
In more recent times graphical representations have also been very
successful in the study of quantum models, not only the Ising
model~\cite{akn,bjo_irb,bjogr,campanino,driessler} 
but also Heisenberg 
models~\cite{aizenman_nacht,toth,ueltschi}
(see also~\cite{golds}).
The transverse-field
Ising model possesses at least three graphical representations,
which may be called firstly the space--time spin representation, 
secondly the random-parity representation, and
thirdly the {\fk}-representation.  The first
of these goes back  to~\cite{driessler}, the
last to to~\cite{akn,campanino},
whereas the second was developed in~\cite{bjogr} 
(see also~\cite{crawford_ioffe}
for the related random-current representation).  These graphical
representations are obtained by applying the Lie--Trotter expansion to
the correlations functions~\eqref{corr_eq} using the eigenbasis for either
$\s^{(3)}$ or $\s^{(1)}$, see~\cite{ioffe_geom}.  

Of primary importance for the
present work is the random-parity representation, which is described
in Section~\ref{rpr_sec}.  It is a continuous version of the 
random-current representation for the classical Ising model, developed 
in~\cite{aiz82}.  The key insight of the works~\cite{aiz82,adcs}
was that the phase transition in the (classical) Ising model relates to
a percolation transition for the random currents,
and the present work exploits a similar picture for the
quantum model.

In addition to graphical representations, the other key component of
the proof of Theorem~\ref{main_thm} is an \emph{infrared bound}
proved in~\cite{bjo_irb}, and stated below
in~\eqref{irb}.
This is a bound on the Fourier-transform of the
Schwinger function: 
\begin{equation}\label{schwinger_eq}
c((x,s),(y,t))=
\frac{1}{\tr(e^{-\b H_\L})}
\tr(e^{-(\b-t+s) H_\L}\s_y^{(3)}e^{-(t-s) H_\L}\s_x^{(3)}).
\end{equation}
Infrared bounds go back to~\cite{fss}
and one of their great successes is the proof by Dyson, Lieb and
Simon~\cite{dls} of the existence of a phase transition
in the antiferromagnetic Heisenberg model.  
See also~\cite{ueltschi} for a recent infrared
bound for the Heisenberg model in the same spirit as the bound
employed here.

The argument for proving Theorem~\ref{main_thm}
follows the general outline of the argument for the
classical model given in~\cite{adcs}.  
The first step is to develop an infinite-volume
version of the random-parity 
representation and study percolation under this
measure, see Section~\ref{infvol_sec}.  The infrared bound is 
used to show that when $\l=\l_\crit$ then
(for $\b<\oo$ and $d\geq3$ or $\b=\oo$ and $d\geq2$)
there is no unbounded percolation cluster,
see Proposition~\ref{crit-perc_prop}.  
Combined with `local modifcations' of the random-parity 
representation (Proposition~\ref{local_prop}) and the
switching lemma (Lemma~\ref{sw_lem}), this allows us 
to deduce the result, as described in Section~\ref{proof_sec}.
Compared with the classical model~\cite{adcs}, 
the main difficulty in the present work arises from the
`continuous' nature of the graphical representations in the quantum
setting.  For example, the configuration space of
the random-parity representation is non-compact so an argument
is needed to obtain tightness of the sequence of finite-volume 
random-parity measures (see Proposition~\ref{weak_prop}).
Related difficulties arise in proving insertion  tolerance
(see Proposition~\ref{local_prop}) and ergodicity 
(see Lemma~\ref{mix_lem}).

In the rest of this article we
fix $\d>0$.  We also set $\g=0$ and use the equality
of the residual and spontaneous 
magnetization~\eqref{spont_resid}.
We  use the abbreviation {\tfim} for transverse-field Ising
model, and we use the following probabilistic notation:
$\one_A$ or $\one\{A\}$ for the indicator function
taking value 1 if the event $A$ occurs, 0 otherwise, and
$\PP(X)$ for the expectation of the random variable $X$
under the probability measure $\PP$.

\section{Graphical representations}
\label{graphical_sec}

In this section we present two graphical representations of the
{\tfim}, namely the space--time spin
representation and the random-parity representation.  The latter
represents the correlation functions~\eqref{corr_eq} 
and Schwinger functions~\eqref{schwinger_eq} 
in terms of expectations over random
`paths' and is of central importance to this work.  A major technical
tool in this representation is the \emph{switching lemma}, which we
describe in Section~\ref{switch_sec}.  In Section~\ref{local_sec}
we then prove some results on `local
modifications' in this representation, which are forms of insertion-
and deletion tolerance.

The space--time spin representation 
has a less prominent role in the
main argument than the random-parity representation,
and is used mainly to
establish an ergodicity property of the infinite-volume random-parity
representation
(see Lemma~\ref{mix_lem} and Proposition~\ref{ergodic_prop}).  
However, it provides a natural setting to
introduce a class of boundary conditions that are of central
importance to this work, and may
also provide a more intuitive description
of the {\tfim} to readers with a
probabilistic background than the definition given in
Section~\ref{intro_sec}.

Before proceeding we introduce some notation.  
Recall that $\L_N$ denotes the box $[-N,N]^d\se\ZZ^d$ and
let $\partial\L_N=\L_N\sm\L_{N-1}$ denote the boundary
of $\L_N$.
For $r>0$ write $I_r$ for the interval $[-r/2,r/2]\se\RR$ and
define $K(N,r)=\L_N\times I_r$.  Also write 
$I_\oo=\RR$, $\KK_\b=\ZZ^d\times I_\b$ 
and $\KK=\KK_\oo=\ZZ^d\times I_\oo$. 
We frequently think of $\KK$ as a subset of $\RR^{d+1}$
and $K(N,r)$ as a subset of $\KK$ in the natural way.
For elements $x\in\ZZ^d$ or $(x,t)\in\KK$ we write $\|x\|$ and
$\|(x,t)\|=\|x\|+|t|$ for their $\ell^1$-norm.  We
will also use the notation $\cE=\{xy:x,y\in\ZZ^d,\|x-y\|=1\}$
for the set of unordered pairs of nearest neighbours in $\ZZ^d$,
$\cE_N=\{xy:x,y\in\L_N,\|x-y\|=1\}$
for  nearest neighbours in $\L_N$,
and $F(N,r)=\cE_N\times I_r$ as well as 
$\FF=\cE\times I_\oo$.  

\subsection{Space--time spin representation}
\label{spin_sec}

Let $\S_\b$ be the set of functions 
$\s(\cdot,\cdot):\KK_\b\to\{-1,+1\}$ such that for all 
$x\in\ZZ^d$, the restriction $\s(x,\cdot):I_\b\to\{-1,+1\}$
is right-continuous and changes value finitely often
in each bounded interval.  Also let $\S=\S_\oo$
and let $\S(N,r)$ be
the set of restrictions of elements of $\S$ to $K(N,r)$.
The space--time spin representation is
based on probability measures on $\S(N,r)$.  
To define these, let $E$ be a 
probability measure governing
\begin{itemize}
\item[(a)]  a collection $D=(D_x:x\in\ZZ^d)$ of independent 
Poisson processes on $\RR$ of rate $\d$, and
\item[(b)]  a collection $\xi=(\xi_x:x\in\ZZ^d)$
of independent random variables taking the values
0 or 1 with equal probability.  
\end{itemize}
We let 
\begin{equation*}
\s(x,t)=\left\{
\begin{array}{ll}
  (-1)^{\xi_x+|D_x\cap (0,t]|} & \mbox{if } t\geq0,\\
(-1)^{\xi_x+|D_x\cap (t,0]|} & \mbox{if } t<0.
\end{array}\right.
\end{equation*}
Thus $\s(x,0)=(-1)^{\xi_x}$ and $\s(x,t)$ switches
value at the points $t\in D_x$.  See Figure~\ref{graphical_fig}
for an example.
We sometimes write $\s_{x,t}$ for $\s(x,t)$.
In this way $E$
defines an `a-priori' measure on $\S$.  We write $E_{N,r}$ for the
induced measure on $\S(N,r)$.

The general notation for space--time Ising 
probability measures on $\S(N,r)$ will be of
the form 
\begin{equation*}
\mu_{N,r}^{\mathfrak{s},\ft}(\cdot),\qquad\mbox{where }
\mathfrak{s},\ft\in\{\f,\w,\p\}.
\end{equation*}
The superscripts $\mathfrak{s},\ft$ 
denote boundary conditions, `spatial' and 
`temporal', respectively, and their
values $\f,\w,\p$ stand for `free', `wired' and 
`periodic'.  The easiest to describe is
$\mu_{N,r}^{\f,\f}$, which is given by its density
\begin{equation}\label{strn_eq}
\frac{d\mu_{N,r}^{\f,\f}}{dE_{N,r}}(\s)=\frac{1}{Z^{\f,\f}_{N,r}}
\exp\Big(\l\sum_{xy\in\cE_N}\int_{I_r}\s(x,t)\s(y,t)dt\Big).
\end{equation}
Here 
\begin{equation}\label{stpf_eq}
Z^{\f,\f}_{N,r}=E_{N,r}\Big[\exp\Big(
\l\sum_{xy\in\cE_N}\int_{I_r}\s(x,t)\s(y,t)dt\Big)\Big]
\end{equation}
is the appropriate normalization.  To obtain the 
boundary conditions 
$\ft=\w$ and $\ft=\p$, 
respectively, we include in the
density~\eqref{strn_eq} (and in the partition
function~\eqref{stpf_eq}) the following restrictions:
\begin{itemize}
\item for $\ft=\w$, that $\s(x,-r/2)=\s(x,r/2)=+1$ for all $x\in\L_N$, 
\item for $\ft=\p$, that $\s(x,-r/2)=\s(x,r/2)$ for all $x\in\L_N$.
\end{itemize}
(Thus, intuitively, for $\ft=\w$ the spins at
the endpoints of $I_r$ are `frozen' to be $+1$, whereas
for $\ft=\p$ we think of $I_r$ as a circle
of length $r$.)
To obtain the 
boundary condition $\mathfrak{s}=\w$
we  replace $\cE_N$
by $\cE_{N+1}$ in the sums 
in~\eqref{strn_eq}--\eqref{stpf_eq} and set
$\s(x,t)=+1$ for any $x\in\partial\L_{N+1}$.
Finally, to obtain the 
boundary condition $\mathfrak{s}=\p$ 
we replace $\cE_N$ by the set $\cE_N^\p$ obtained by adding
to $\cE_N$ all pairs $xy$ such that $x$ and $y$  differ in exactly one 
coordinate, this coordinate being $-N$ in one case and $N$ in the
other.  (Intuitively this makes $\L_N$ `wrap around' in each
coordinate direction.) 

\begin{figure}[tbp]
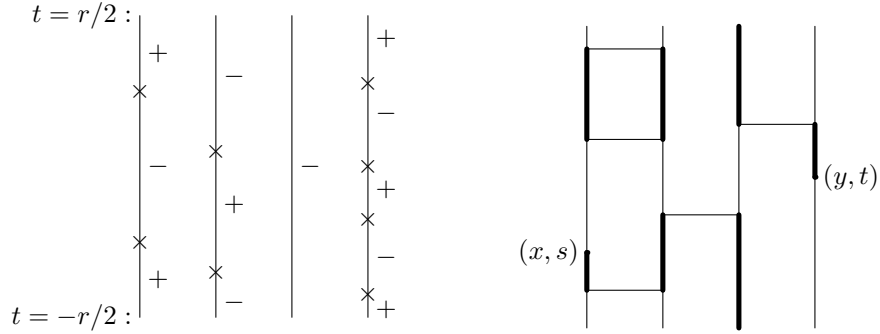

\includegraphics{critmag.3}
\hspace{1.3cm} 
\includegraphics{critmag.4}
\caption
{\emph{Left:} A realization of the spin-representation $\s$.
The elements of $D$ are marked as $\times$, and the value of $\s$ is
indicated next to each line segment.  (This realization is compatible
with $\ft=\f$ or $\p$ but not with $\ft=\w$.)  
\emph{Right:} A realization of the random-parity 
representation $\psi^\p_A$, where $A=\{(x,s),(y,t)\}$.  `Odd'
intervals are drawn bold, bridges as horizontal lines.  (This
realization is not compatible with $\ft=\f$ or $\w$.)}
\label{graphical_fig}
\end{figure}

The connection to the {\tfim} is that the correlation
functions~\eqref{corr_eq} satisfy
\[
\langle \s_A^{(3)}\rangle_{\L,\b}=
\mu_{n,\b}^{\f,\p}\Big(\prod_{x\in A}\s(x,0)\Big), 
\]
and the Schwinger function~\eqref{schwinger_eq}
\[
c((x,s),(y,t))=\mu_{n,\b}^{\f,\p}\big(\s(x,s)\s(y,t)\big).
\]
(Recall that we write $\mu(f)$ for the
$\mu$-expectation of a function $f(\s)$.)  
In light of this correspondence it is natural to use the notation 
\begin{equation}\label{corr_eq2}
\langle \s_A\rangle_{N,r}^{\mathfrak{s},\ft}=
\mu_{N,r}^{\mathfrak{s},\ft}
\Big(\prod_{(x,t)\in  A}\s(x,t)\Big)
\end{equation}
for general finite subsets $A\se K(N,r)$.  Henceforth we refer to the
quantities in~\eqref{corr_eq2} as \emph{correlation functions}.

The $\p$-boundary condition
in `time' arises automatically from the cyclicity of the trace
in~\eqref{Q-state}.  The other boundary conditions 
$\ft=\f,\w$ are
convenient when working with infinite-volume limits.  
The correlation functions~\eqref{corr_eq2} have
a natural monotonicity in the boundary conditions.
In particular, 
\begin{equation}\label{corr_comp}
\langle \s_A\rangle_{N,r}^{\f,\f}\leq
\langle \s_A\rangle_{N,r}^{\f,\p}\leq
\langle \s_A\rangle_{N,r}^{\p,\p}\leq
\langle \s_A\rangle_{N,r}^{\w,\p}\leq
\langle \s_A\rangle_{N,r}^{\w,\w},\mbox{ etc.}
\end{equation}
(A detailed proof  for the present model 
may easily be devised using Theorem~2.2.12 and
Lemma~2.2.21 of~\cite{bjo_phd}.)

When $\b<\oo$ we fix $\ft=\p$ 
and work with the measures
$\mu_{N,\b}^{\f,\p}$ and $\mu_{N,\b}^{\w,\p}$.  When $\b=\oo$
we will primarily be working
with $\mu_{N,r}^{\f,\f}$ and $\mu_{N,r}^{\w,\w}$
where $r=2N$.  The (weak) limits of these measures 
exist as $N\to\oo$ and are related,
respectively, to the positive-temperature and ground-state limits
appearing in~\eqref{resid_eq}.  In particular, 
\begin{equation}\label{spont_resid}
M^+_\oo=\lim_{N\to\oo} \langle \s(0,0)\rangle_{N,r=2N}^{\w,\w},\quad
M^+_\b=\lim_{N\to\oo} \langle \s(0,0)\rangle_{N,\b}^{\w,\p}
\mbox{ for } \b<\oo,
\end{equation}
i.e.\ the residual and spontaneous magnetization
coincide, cf.~\cite{lml} and Section~2.5.2 of~\cite{bjo_phd}.
We write $\mu_{\l,\b}^{(\mathfrak{s},\ft)}(\cdot)$ 
or $\el\cdot\er_{\l,\b}^{(\mathfrak{s},\ft)}$ for any weak
limit obtained as above with boundary conditions 
$\mathfrak{s},\ft$ (and $\b<\oo$ or $\b=\oo$).
When $M^+=0$ then the limit measure is unique:
\[
\mu_{\l,\b}^{(\f,\p)}=\mu_{\l,\b}^{(\w,\p)}\mbox{ for }\b<\oo,\qquad
\mu_{\l,\oo}^{(\f,\f)}=\mu_{\l,\oo}^{(\w,\w)}.
\]
(A detailed proof for the present model appears 
in~\cite[Theorem~2.5.9]{bjo_phd};  it follows closely the
argument for the classical model~\cite{lml}.)
Thus our main result Theorem~\ref{main_thm} implies that the limit
measure is unique at the critical point.
For more information about the statements
in this subsection, and the spin-representation in 
general, see~\cite{bjo_phd}.

\subsection{Random-parity representation}
\label{rpr_sec}

Like the space--time spin representation of the 
previous subsection, the random-parity representation
expresses the correlation functions~\eqref{corr_eq2}
using Poisson processes in $K(N,r)$.
This time we write $E_{N,r}$ for a probability measure
governing a collection $B=(B_{xy}:xy\in\cE_N)$
of independent Poisson processes on $I_r$ with
intensity $\l$, as well as an independent collection
$\tau=(\tau_x:x\in\L_N)\in\{0,1\}^{\L_N}$ 
of independent random variables taking values 0 or 1
with equal probability $1/2$.  We sometimes refer to the points 
of $B$ as \emph{bridges}.
When considering the correlation 
function $\langle \s_A\rangle_{N,r}^{\mathfrak{s},\ft}$
we will use the term \emph{switching point}
for any point $(x,t)\in K(N,r)$
such that either (i) $(x,t)\in A$, or (ii)
there exists $y\in\L_N$ such that $t\in B_{xy}$
(ie, $(x,t)$ is the `endpoint' of some bridge).
The collection of switching points is denoted by 
$S=(S_x:x\in\L_N)$, where $S_x$ is the set of $t\in I_r$
such that $(x,t)$ is a switching point.  
The points of $A$ are referred to as \emph{sources}.
We say that  $B$
is \emph{consistent} with $A$ if $|S_x|$ is
even for each $x$, and in this case we will also refer
to $S$ itself as consistent.

For each consistent $S$ we will define a labelling of
$K(N,r)$ using the labels `even' and `odd', and
for definiteness we use the convention that the
`odd' subset is closed.  See Figure~\ref{graphical_fig} again for an
illustration of the description that follows.
The definition  will depend on the boundary 
conditions $\mathfrak{s}$
and $\ft$.  We will not be using the
random-parity representation for 
$\mathfrak{s}=\p$ so we omit describing it.
For $\mathfrak{s}=\f$ we will denote the
labelling $\psi^{\ft}_A$ and for
$\mathfrak{s}=\w$ by 
$\hat\psi^{\ft}_A$.  
In what follows we assume for simplicity
that $A$ does not contain any point
of the form $(x,\pm r/2)$.  We begin with the case
$\mathfrak{s}=\f$:
\begin{list}{$\bullet$}{\leftmargin=1em}
\item If  $\ft=\f$ 
we label each point
$(x,\pm r/2)$ `even', and define the rest of 
$\psi^\f_A$
by requiring that labels switch between `even'
and `odd' at the points of $S$ and are constant 
in between.  This is possible
due to the assumption that $S$ is consistent.
\item If $\ft=\w$ we instead
label all points of the form $(x,\pm r/2)$ `odd' and
apply the same rule for switching at points of $S$.  
\item If $\ft=\p$ we require the vector $\tau$.
We define $\psi_A^\p$
by labelling each $(x,0)$ `even' if $\tau_x=0$
or `odd' if $\tau_x=1$, and letting the label switch
at the points of $S$ as before.  Due to the consistency
of $S$ we can think of $\psi^\p_A$ as a labelling
of the circle.  (Clearly the choice $(x,0)$
is arbitrary, and one may equally well let
$\tau$ determine all the labels $(x,t)$ for any 
fixed $t$.)
\end{list}

We now describe how to define the labelling in the 
case when $\mathfrak{s}=\w$.  We now let 
$\hat E_{N,r}$ denote a probability measure
which, in addition to processes $B$ and $\tau$
as above, also governs a process 
$G=(G_x:x\in\partial\L_N)$ of independent Poisson
processes on $I_r$.  The intensity of $G_x$ depends
on $x$, and equals $\l$ times the number of 
$y\in\ZZ^d\sm\L_N$ such that $xy\in\cE$.  
For simplicity we set $G_x=\es$ for
$x\in\L_N\sm\partial\L_N$.
We augment the switching points $S$ to 
contain also each $(x,t)$ such that $t\in G_x$.
As before we say that $S$ is consistent if each
$|S_x|$ is even, and in that case we also say that the
pair $(B,G)$ is consistent with $A$.
Given a consistent $S$ we obtain the labelling, which we
now denote 
$\hat\psi^\ft_A$, precisely as before.

Given a labelling $\psi^\ft_A$ 
(respectively, $\hat\psi^\ft_A$), 
let $\epsilon$ denote
the total length of all intervals in $K(N,r)$
labelled `even', and let the
\emph{weight} $\partial\psi^\ft_A$ 
(respectively, $\partial\hat\psi^\ft_A$) 
be defined as $e^{2\d\epsilon}$.  In the case when
$S$ is not consistent we define the weight to be zero.
The random-parity
representation allows us to write
\begin{equation}\label{rpr_eq}
\langle \s_A\rangle^{\f,\ft}_{N,r}=
\frac{E_{N,r}(\partial\psi^\ft_A)}
{E_{N,r}(\partial\psi^\ft_\es)}
\quad \mbox{and} \quad
\langle \s_A\rangle^{\w,\ft}_{N,r}=
\frac{\hat E_{N,r}(\partial\hat\psi^\ft_A)}
{\hat E_{N,r}(\partial\hat\psi^\ft_\es)}.
\end{equation}
The first of these ($\mathfrak{s}=\f$) was proved 
in~\cite[Theorem~3.1]{bjogr}, and the second ($\mathfrak{s}=\w$)
follows using similar arguments.

In fact, the representation~\eqref{rpr_eq}
holds for correlations in more general subsets
$K'$ of $\KK$ than $K(N,r)$, in particular
for `regions with holes'.  We briefly outline this 
now (it will be used in Proposition~\ref{weak_prop}).
Write $K=K(N,r)$, let $J$ be a finite collection
of disjoint closed intervals in $K$, and let $K'=K\sm J$.
Write $\partial J$ for the set of all endpoints
of intervals in $J$ (if $J$ contains some interval
consisting of a single point $(x,t)$ we
distinguish between $(x,t+)$ and $(x,t-)$).
For each $xy\in\cE_N$ let $B_{xy}$ denote a Poisson
process on $I_r$ with variable intensity:
$\l$ if $(x,t),(y,t)\in K'$, otherwise 0.  We 
start by labelling
each point $(x,t)\in\partial J$ `even'.  If $\ft=\f$
(respectively, $\ft=\w$) we label each point
$(x,\pm r/2)\in K'$ `even' (respectively, `odd').
If $\ft=\p$ then for each $x$ such that 
$\{x\}\times I_r\se K'$ we let $\tau_x$ determine
the label $(x,0)$ as before.  The full labelling is finally
obtained by letting
the labels switch at the points of $S$ as before.
(A labelling $\hat\psi_A^\ft$ of $K'$ is 
obtained similarly, modifying also $G$ to have
zero intensity in $J$.)  Writing $E_{K'}$ for the
corresponding probability measure, we note that
\begin{equation}\label{holes_eq}
E_{K'}(\partial\psi_\es^\ft)=
e^{-2\d |J|}
E_{N,r}(\partial\psi_\es^\ft
\one\{\psi_\es^\ft\mbox{ `even' in }J\}),
\end{equation}
where $|J|$ denotes the total
length of the intervals comprising $J$,
see~\cite{bjogr}.

\subsubsection{Switching lemma}
\label{switch_sec}

Consider the case in~\eqref{rpr_eq} 
when $A$ consists of the two points $(0,0)$
and $(x,t)$.  The consistency constraint on $S$, and 
the way the labels are defined, forces the existence 
of an `odd path' between the two sources $(0,0)$
and $(x,t)$.  A similar, but more complicated,
picture arises for general sets $A$.  The main virtue
of the random-parity representation is that this
picture can be developed to the case of
\emph{pairs} of labellings in a way which also allows the
representation of \emph{differences} between correlation
functions.  The tool for this is called the
switching lemma.

Let $\EE_{N,r}$ denote a probability measure governing
the following independent random variables:
\begin{itemize}
\item[(a)] two copies of the process
$B$, denoted $B$ and $\hat B$;
\item[(b)] two  copies of $\tau$, 
denoted $\tau$ and $\hat\tau$;  
\item[(c)] one copy of the process $G$;  and
\item[(d)] one copy of a process $\D=(\D_x:x\in\L_N)$,
where the $\D_x$ are independent Poisson processes
on $I_r$ with intensity $4\d$.
\end{itemize}
Thus we may write $\EE_{N,r}=E_{N,r}\times\hat E_{N,r}\times P_\D$,
where $P_\D$ denotes the distribution of $\D$.
We call $\D$ the process of `cuts'.
If we fix two finite sets $A_1$ and $A_2$ of
sources and two boundary conditions $\ft_1$
and $\ft_2$
then we obtain under $\EE_{N,r}$  a triple 
$(\psi_{A_1}^{\ft_1},\hat\psi_{A_2}^{\ft_2},\D)$ whose
components are independent.  
In what follows we will only be using the following 
combinations of boundary conditions $\ft_1$, $\ft_2$:
in the case $\b<\oo$ we take $\ft_1=\ft_2=\p$,
and in the case $\b=\oo$ we take $\ft_1=\f$
and $\ft_2=\w$, cf.\ the discussion below~\eqref{corr_comp}.
In the latter case we use the
notation $\hat G_x=G_x\cup \{-r/2,r/2\}$.

We next define a notion of
\emph{open paths}.  The precise definition of
a path depends on the boundary 
conditions $\ft_1,\ft_2$, and we start
with the case $\ft_1=\f,\ft_2=\w$.  
If $\k,\k'\in K(N,r)$
then an (open) path from $\k$ to $\k'$ is a sequence
of points $\k_0,\k_1,\dotsc,\k_{2m+1}\in K(N,r)$
satisfying the following:
\begin{enumerate}
\item $\k_0=\k$ and $\k_{2m+1}=\k'$;
\item\label{int_item}  
for each $j\in\{0,\dotsc,m\}$, if $\k_{2j}=(x,s)$
then $\k_{2j+1}=(x,t)$ for some $t$ such that
the interval $[s\wedge t,s \vee t]$ contains no
point $s'\in\D_x$ which is labelled `even'
in both labellings $\psi_{A_1}^\f$ and
$\hat\psi_{A_2}^\w$;
\item\label{jump_item} 
for each $j\in\{0,\dotsc,m-1\}$,
if $\k_{2j+1}=(x,s)$ and $\k_{2j+2}=(y,t)$
then either (i) $s=t\in B_{xy}\cup \hat B_{xy}$,
or (ii) $s\in\hat G_x$ and $t\in\hat G_y$.
\end{enumerate} 
Intuitively this  means that paths can
traverse bridges, `jump between' arbitrary
points of $\hat G$, and traverse subintervals of $K(N,r)$,
but are blocked by points of $\D$ which 
fall where both labellings are `even'.  
One way to think of the `jumping'
between points of $\hat G$ is that 
the points in $\hat G$ are connections to and
from a `ghost-site' $\G$.

To obtain the case $\ft_1=\ft_2=\p$
we modify (ii) in item~\eqref{jump_item} above
by replacing $\hat G$ with $G$ in both places,
and we additionally allow in~\eqref{int_item}
that one may traverse either the inverval
$[s,t]$ or the interval $[t,s]$, where these
are to be regarded as intervals in the circle
(we keep the restriction on $\D$).  
Intuitively these changes mean that
the `endpoints' $(x,\pm r/2)$ are no longer
connected to $\G$ but are identified
with each other.

The event that there is an open path between $\k$
and $\k'$ is written $\{\k\lra\k'\}$.  The 
possibility (ii) in~\eqref{jump_item} of
`jumping' via $\G$ is special, and in some
cases we want to consider paths which do not
do this.  We write 
$\{\k\lra\k'\mbox{ off }\G\}$ for the event that
there is some path that does not 
feature a pair $\k_{2j+1}=(x,s)$, $\k_{2j+2}=(y,t)$
with $x\neq y$ and $s,t\in\hat G$.
Note that when $\ft_1=\f,\ft_2=\w$ this also excludes
jumping via the endpoints $(x,\pm r/2)$.  
We also write $\{\k\lra\G\}$ for the event that
some open path connects $\k$ to a point in $\hat G$
(respectively, $G$).

Write 0 for the origin $(0,0)$ and
$\k$ for an arbitrary point in $K(N,r)$.
We will be using the following form of the 
switching lemma:
\begin{lemma}\label{sw_lem}
$\EE_{N,r}(\partial\psi^{\ft_1}_{0\k}
\partial\hat\psi_\es^{\ft_2})=
\EE_{N,r}(\partial\psi_\es^{\ft_1}\partial\hat\psi_{0\k}^{\ft_2}
\one\{0\lra \k\mbox{ off }\G\})$.
\end{lemma}
Here we have abbreviated $\{0,\k\}$
with $0\k$.
The proof of Lemma~\ref{sw_lem} 
is a small modification 
of~\cite[Theorem~4.2]{bjogr}, 
which we briefly outline now.  
\begin{proof}[Proof sketch]
Note that
in the left-hand-side, the event 
$\{0\lra x\mbox{ off }\G\}$ holds since there is an
odd path in $\psi_{0\k}^{\ft_1}$.  Condition on the sets
$\overline B=B\cup\hat B$ and $G$, and note that the
conditional distribution of the pair $(B,\hat B)$
is given by assigning each element of $\ol B$
to $B$ or $\hat B$ with equal probability, independently.  
Given $\ol B$ and $G$ there is a finite collection of 
`possible' paths $\pi_1,\dotsc,\pi_n$ between $0$ and $\k$
off $\G$ (the numbering is arbitrary but fixed,
and as noted above the collection is nonempty).  
When we assign the 
elements of $\ol B$ to $B$ and $\hat B$ and also sample
$\D$ and (in the case $\ft_1=\ft_2=\p$) $\tau,\hat\tau$,
this will both determine the labellings 
$\psi^{\ft_1}_{0\k}$ and $\hat\psi_\es^{\ft_2}$, as
well as reveal which of the possible paths are indeed open.
Let $\pi_j$ be the first of these.  Write 
$\psi^{\ft_1}_{0\k}\sd \pi_j$ and 
$\hat\psi_\es^{\ft_2}\sd\pi_j$ for the labellings obtained
by switching `even' and `odd' along $\pi_j$.  
It is easy to see that these new labellings
are consistent with sources $\es$ and $0\k$,
respectively, and that they can be obtained
by switching certain elements between $B$ and $\hat B$ as well
as certain values of $\tau$ and $\hat\tau$.  By symmetry,
the latter transformations are measure-preserving.  Moreover,
it may be checked as in~\cite[p.~251]{bjogr} that the
change in \emph{weight} of the labellings exactly corresponds
to the change in \emph{probability} that $\pi_j$ 
is indeed the first open path.  In particular,
in the new configuration there is still a path
between 0 and $\k$ off $\G$.
\end{proof}

The following simple application of 
Lemma~\ref{sw_lem} will be used in the proof of
Theorem~\ref{main_thm}.     We have that
\begin{equation}\label{sw2}
\begin{split}
\EE_{N,r}(\partial\psi_\es^{\ft_1}\partial\hat\psi_{0\k}^{\ft_2})
-\EE_{N,r}(\partial\psi_{0\k}^{\ft_1}\partial\hat\psi_\es^{\ft_2})
&=\EE_{N,r}(\partial\psi_\es^{\ft_1}\partial\hat\psi_{0\k}^{\ft_2}
[1-\one\{0\lra \k\mbox{  off }\G\}])\\
&\leq 
\EE_{N,r}(\partial\psi_\es^{\ft_1}\partial\hat\psi_{0\k}^{\ft_2}
\one\{0\lra \G\}).
\end{split}
\end{equation}
In the last step we used that $\hat\psi_{0\k}^{\ft_2}$
contains an odd path $\pi$ between 0 and $\k$, and if there
is no path between 0 and $\k$ off $\G$ then $\pi$
must pass $\G$.
Combined with~\eqref{rpr_eq} this has the following
consequence.    Writing $0$ for $(0,0)$
and $\k$ for $(x,t)$ we have from~\eqref{sw2}
\begin{equation}\label{tp_rpr}
\begin{split}
0\leq 
\langle \s_{0,0}\s_{x,t} \rangle^{\w,\ft_2}_{N,r}
-\langle \s_{0,0}\s_{x,t} \rangle^{\f,\ft_1}_{N,r}
&\leq \frac{
\EE_{N,r}(\partial\psi_\es^{\ft_1}\partial\hat\psi_{0\k}^{\ft_2}
\one\{0\lra \G\})}
{\EE_{N,r}(\partial\psi_\es^{\ft_1}\partial\hat\psi_\es^{\ft_2})}.
\end{split}
\end{equation}

\subsubsection{Local modifications}
\label{local_sec}

Of central importance to the proof of 
Theorem~\ref{main_thm} is a probability measure
$\ol\PP_{N,r}$ defined as follows.  Recall the
processes $B,\hat B, G, \D,\tau, \hat \tau$
as well as the labellings $\psi,\hat\psi$
of the previous subsection.  If $A$ is an event
measurable with respect to these processes, we define
\begin{equation}\label{P_bar_def}
\ol\PP_{N,r}(A)=
\frac{\EE_{N,r}(\one_A\partial\psi_\es^{\ft_1}\partial\hat\psi_\es^{\ft_2})}
{\EE_{N,r}(\partial\psi_\es^{\ft_1}\partial\hat\psi_\es^{\ft_2})}.
\end{equation}
In this section we prove two estimates relating
to $\ol\PP_{N,r}$.  Before we state these, recall
our conventions:  if $\b<\oo$ then we write $r=\b$
and let $\ft_1=\ft_2=\p$, whereas if $\b=\oo$
then $r=2N\to\oo$ and $\ft_1=\f,\ft_2=\w$.  In the following
result we let $N_0<N$ and if $\b<\oo$ let $r_0=r=\b$,
or if $\b=\oo$ let $r_0=2N_0<2N=r$.  If $J$
is a measurable subset of $K(N,r)$ we say that the
event $A$ is \emph{defined in} $J$ if it is
measurable with respect to the restrictions of the 
processes $B,\hat B, G, \D,\psi,\hat\psi$ to $J$.
\begin{proposition}\label{local_prop}
\hspace{1cm}\\
(A) For each $(x,t)\in K(N,r)$ there is a constant 
$C_{(x,t)}$ not depending on $N,r$ such that
\[
\langle \s_{0,0}\s_{x,t} \rangle^{\w,\ft_2}_{N,r}
-\langle \s_{0,0}\s_{x,t} \rangle^{\f,\ft_1}_{N,r}
\leq C_{(x,t)}\ol\PP_{N,r}((0,0)\lra \G).
\]\\
(B) There is a constant $c=c(N_0,r_0)$ such that the
following holds.
Let $A$ be an event defined in $K(N,r)\sm K(N_0,r_0)$,
and let $\cC$ be the event that there is an open
path inside $K(N_0,r_0)$ between every pair of 
points in $K(N_0,r_0)$.  Then
\[
\ol\PP_{N,r}(A)\leq c \,\ol\PP_{N,r}(A\cap\cC).
\]
\end{proposition}

In proving Proposition~\ref{local_prop}
we will be using the following  fact about
`local modifications' of point processes.  Let
$X$ denote a point process 
on the interval $[0,t]$.  Let $\tilde X$ be another
point process on $[0,t]$ obtained from $X$
by a deterministic or random modification.  For example,
$\tilde X$ may be obtained by adding a point somewhere in
$[0,t]$, or deleting one of the points of $X$.
Write $E,\tilde E,\EE$ for the law of $X$, the
law of $\tilde X$ and their joint law, respectively.
We will assume that $\tilde X$ is defined in such 
a way that for some event $A$ we have
$\tilde X\in A$ with probability 1.  Moreover, assume
that $f$ is some function and $c_1,c_2>0$ some constants
such that $\tfrac{f(X)}{f(\tilde X)}\leq c_1$ and
the Radon-Nikodym density $\tfrac{d\tilde E}{d E}\leq c_2$
almost surely.  We then have that
\begin{equation}\label{modif_bound}
\begin{split}
E[f(X)]&= \EE[f(X)]=\EE[f(X)\one\{\tilde X\in A\}]
\leq c_1 \EE[f(\tilde X)\one\{\tilde X\in A\}]\\
&= c_1 \tilde E[f(\tilde X)\one_A(\tilde X)]=
c_1 E\Big[\frac{d \tilde E}{d E}(X)f(X)\one_A(X)\Big]\\
&\leq c_1c_2 E[f(X)\one\{X\in A\}].
\end{split}
\end{equation}
We will be using~\eqref{modif_bound} 
when $X$ is a Poisson process of intensity
$\a$, say, and $\tilde X$ is obtained in one of 
the following three ways.
\begin{list}{$\bullet$}{\leftmargin=1em}
\item Firstly, if $\tilde X$ is
the trivial process obtained by deleting all
points of $X$ then
\begin{equation}\label{delete_RN}
\frac{d\tilde E}{d E}(X)=e^{\a t}\one\{X=\es\}
\leq e^{\a t}.
\end{equation}
\item Secondly, suppose $\tilde X$ is obtained from $X$
by adding two independent, uniformly distributed
points if $X=\es$, but letting $\tilde X=X$ 
otherwise.  Then
\begin{equation}\label{add_0or2_RN}
\begin{split}
\frac{d\tilde E}{d E}(X)&=
\one\{X\neq\es\}+\tfrac{2}{(\a t)^2}\one\{|X|=2\}\\
&\leq 1+\tfrac{2}{(\a t)^2}.
\end{split}
\end{equation}
\item Thirdly, suppose $\tilde X$ is obtained from $X$
by adding uniformly a point if $|X|\in\{0,1\}$, 
alternatively deleting a uniformly chosen point if $|X|\geq2$.
Then 
\begin{equation}\label{add_del_RN}
\begin{split}
\frac{d\tilde E}{d E}(X)&=
\tfrac{1}{\a t}\one\{|X|=1\}
+\tfrac{2}{\a t}\one\{|X|=2\}+
\one\{X\neq\es\}\tfrac{\a t}{|X|+1}\\
&\leq \tfrac{2}{\a t}+\a t.
\end{split}
\end{equation}
\end{list}
One way to check~\eqref{delete_RN}--\eqref{add_del_RN}
is to approximate $X$ by a Bernoulli process 
on $\{0,\tfrac{1}{n},\dotsc,\tfrac{\lfloor tn\rfloor}{n}\}$
with success probability $\a/n$ and look at the
limits of the corresponding likelihood ratios.

\begin{proof}[Proof of Proposition~\ref{local_prop}]
We begin with the easier part (B).  The statement is 
equivalent to
\begin{equation}\label{in1}
\EE_{N,r}(\one_A\partial\psi_\es^{\ft_1}\partial\hat\psi_\es^{\ft_2})
\leq c\,
\EE_{N,r}(\one_A\one_\cC\partial\psi_\es^{\ft_1}\partial\hat\psi_\es^{\ft_2}).
\end{equation}
We modify $\D$ by removing all points in $K(N_0,r_0)$
and we modify $B$ inside $K(N_0,r_0)$ 
as in~\eqref{add_0or2_RN}.
That is, whenever $B_{xy}\cap I_{r_0}=\es$ we
add two bridges uniformly placed in $I_{r_0}$, otherwise
leave $B_{xy}$ unchanged.  The resulting bridge-configuration 
$\tilde B$
is then still consistent with source set $\es$.
If $\b=\oo$ there is (due to our choice of boundary
conditions $\ft_1=\f$, $\ft_2=\w$) a unique labelling
$\tilde \psi_\es^{\ft_1}$ consistent with $\tilde B$ 
which agrees with the original labelling $\psi_\es^{\ft_1}$
in $K(N,r)\sm K(N_0,r_0)$.
If $\b<\oo$ there is a unique
labelling $\tilde \psi_\es^{\ft_1}$ which agrees with 
$\psi_\es^{\ft_1}$ in $K(N,r)\sm K(N_0,r_0)$
and at each $(x,0)$ such that $x\in\L_{N_0}$.
Since we have removed all cuts and placed
bridges between all neighbouring pairs of intervals,
the event $\cC$ holds after the modification.
Since all changes have been restricted to
$K(N_0,r_0)$ the change preserves the event $A$.
The total length $\tilde\epsilon$ of intervals labelled
`even' in $\tilde \psi_\es^{\ft_1}$ satisfies
$\tilde\epsilon\geq\epsilon-r_0(2N_0+1)^d$.
Applying~\eqref{modif_bound} with 
$f$ equal to the weight of the labelling,
as well as~\eqref{delete_RN} and~\eqref{add_0or2_RN}, we 
obtain~\eqref{in1} with
\[
c=\exp(4\d r_0(2N_0+1)^d)
\exp(4\d r_0(2N_0+1)^d)
(1+2/(\l r_0)^2)^{2d(2N_0+1)^d}.
\]
(The first factor is due to the change in 
the weight of the labelling,
the second to the change of measure of $\D$,
and the third to the change of measure of $B$.)

We now turn to part (A).  By~\eqref{tp_rpr}
it suffices to show that
\begin{equation}\label{Cx_eq}
\EE_{N,r}(\partial\psi_\es^{\ft_1}\partial\hat\psi_{0\k}^{\ft_2}
\one\{0\lra \G\})\leq C_\k
\EE_{N,r}(\partial\psi_\es^{\ft_1}\partial\hat\psi^{\ft_2}_\es
\one\{0\lra \G\}),
\end{equation}
where $\k=(x,t)$.
We begin with the (more delicate) case when 
$\b=\oo$ and $r=2N$.  

For each $y\in\L_N$ write
\begin{equation*}
I_{y,k}=\{y\}\times(k,k+1],\qquad -N\leq k\leq N-1.
\end{equation*}
Thus the $I_{y,k}$ form a partition of $K(N,r)$ into
intervals of length 1, and we have that
$(x,t)\in I_{x,\lceil t\rceil -1}$ and 
$(0,0)\in I_{0,-1}$.  We begin by defining a collection 
$\Pi(x,t)$ of intervals of this form which `connect'
$(0,0)$ to $(x,t)$.  There is some flexibility
in the choice of $\Pi(x,t)$, but for definiteness we 
define it as follows.  Firstly, let 
$0=x_0,x_1,\dotsc,x_n=x$ be a fixed, shortest nearest-neighbour
path from $0$ to $x$ in $\L_N$;  thus
$n=\|x\|$.  Next, set $k_0=-1$ and
\begin{equation*}\begin{split}
k_1&=0, k_2=1,\dotsc,k_m=\lceil t \rceil-1\qquad
\mbox{ if } t>0,\\
k_1&=-2, k_2=-3,\dotsc,k_m=\lceil t \rceil-1\qquad
\mbox{ if } t<0.
\end{split}
\end{equation*}
We define
\begin{equation*}
\Pi(\kappa)=\Pi(x,t)=\{I_{0,k_0},I_{0,k_1},\dotsc,I_{0,k_m},
I_{x_1,k_m},I_{x_2,k_m},\dotsc, I_{x_n,k_m}\}.
\end{equation*}
Note that $(0,0)\in I_{0,k_0}$, that
$(x,t)\in I_{x_n,k_m}$, and that the number of intervals
in $\Pi(x,t)$ as well as their total length
are bounded by $|t|+2+\|x\|$.  

We are going to modify $\D$, $\hat B$ and $\hat\psi$
along $\Pi(x,t)$
and apply the argument at~\eqref{modif_bound}. 
We modify $\D$ 
by simply replacing $\D\cap I$ with $\es$ for all
$I\in\Pi(x,t)$.  By~\eqref{delete_RN} the
corresponding Radon--Nikodym density is at most
$\exp(4\d(|t|+2+\|x\|))$.
Next, define
\begin{equation*}
J_i=\{(x_ix_{i+1},s):(x_i,s)\in I_{x_i,k_m},\,
(x_{i+1},s)\in I_{x_{i+1},k_m}\}.
\end{equation*}
We modify $\hat B$ by applying the operation 
in~\eqref{add_del_RN} in each 
$J_i$ for $0\leq i\leq n-1$;
that is, if $J_i$ contains
0 or 1 bridge we add one uniformly, but if $J_i$ contains
2 or more bridges we delete one chosen uniformly.   
Write $\tilde B$ for the modified process of bridges.
By~\eqref{add_del_RN}, the
corresponding Radon--Nikodym density is at most
$(2/\l+\l)^{|t|+2+\|x\|}$.

If $\hat B$ was 
consistent with the sources $0,\k$
then $\tilde B$ is consistent with the source
set $\es$.  Due to the boundary
condition there is a unique labeling $\tilde\psi_\es^{\w}$
associated with $\tilde B$, and in fact
one obtains $\tilde\psi_\es^\w$ from
$\hat\psi_{0\k}^\w$ by modifying the labels in intervals
belonging to $\Pi(x,t)$ only.
See Figure~\ref{local_fig}.
\begin{figure}[tbp]
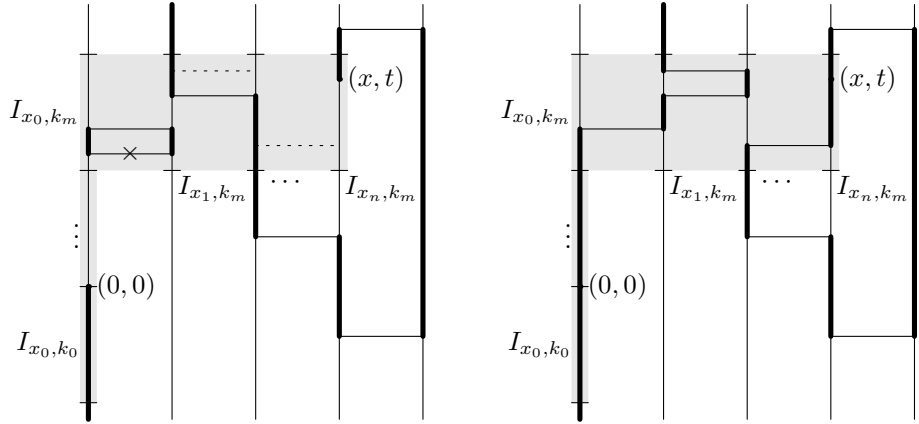

\centering
\includegraphics{critmag.1}
\qquad\includegraphics{critmag.2}
\caption
{Part of the labellings $\hat\psi_{0\kappa}^\w$ (left) and 
$\tilde\psi_{\es}^\w$ (right) in the proof of  
Proposition~\ref{local_prop}(A).  The intervals comprising
$\Pi(\kappa)$ are indicated and highlighted in light
grey.  Intervals labelled `odd' are drawn bold.  Bridges
added (respectively, deleted) are drawn dashed 
(respectively, marked with an $\times$).}
\label{local_fig}   
\end{figure}
It follows that
\begin{equation*}
\frac{\partial\hat\psi_{0\k}^\w}
{\partial\tilde\psi_\es^\w}\leq 
\exp(2\d(|t|+2+\|x\|)).
\end{equation*}
Note that the modifications described above
do not destroy any connections 
(but possibly creates some
new ones).  In particular, if
$0\lra\G$ before
then also $0\lra\G$ after.
Applying the argument in~\eqref{modif_bound}
we therefore arrive at~\eqref{Cx_eq}, with
\[
C_{(x,t)}=\exp(6\d(|t|+2+\|x\|))
(2/\l+\l)^{|t|+2+\|x\|}.
\]

We now turn to the case $\b<\oo$,
and recall that we then have $r=\b$.  We no longer
need to partition $I_\b$ into intervals of length 1,
but instead  define $I_x=\{x\}\times I_\b$.  
We now let 
\begin{equation*}
\Pi(x,t)=\{I_{x_0},I_{x_1},I_{x_2},\dotsc, I_{x_n}\},
\end{equation*}
where as before
$0=x_0,x_1,\dotsc,x_n=x$ is a fixed, shortest, nearest-neighbour
path from $0$ to $x$ in $\L_N$.  The number of 
intervals in $\Pi(x,t)$ is now $n=\|x\|$, and their
combined length is $\b\|x\|$.
With this definition of $\Pi(x,t)$
we apply the same modifications to $\D$ and $\hat B$
as in the case $\b=\oo$.  Now we  
let $\tilde\psi_\es^\p$ be the unique labelling which
agrees with $\hat\psi_{0\k}^\p$ outside the intervals
of $\Pi(x,t)$ and at all points $(x_i,0)$ for
$i=0,1,\dotsc,n$.
This time we get~\eqref{Cx_eq} with
\[
C_{(x,t)}=\exp(6\d\b\|x\|)(2/\l\b+\l\b)^{\|x\|}.
\]
\end{proof}

\section{Infinite-volume RPR}
\label{infvol_sec}

In this section we study the limit $\ol\PP$
of the measures $\ol\PP_{N,r}$ as $N\to\oo$
(and either $r=\b<\oo$ is fixed, or $r=2N\to\oo$).
In Section~\ref{inf_rpr_sec} we prove the existence
of $\ol\PP$ as well as basic properties such
as translation-invariance and ergodicity.  Then in
Section~\ref{perc_sec} we show how the 
argument of Burton and Keane~\cite{burton-keane}
can be adapted to show that, almost surely under
$\ol\PP$, there is either no or exactly one
infinite connected cluster.

\subsection{Existence and basic properties}
\label{inf_rpr_sec}

In proving existence of the limit $\ol\PP$
of the sequence $\ol\PP_{N,r}$ we will 
need to pay attention to
the underlying point processes,
and we will loosely follow the notational
conventions of Daley and Vere-Jones~\cite{dvj}
for point processes.  Recall that the
labelling $\psi_\es^{\ft_1}$ is a function of the
pair $(B,\tau)$ where $B$ is a point process
on $\cE_N\times I_r$ and $\tau\in\{0,1\}^{\L_N}$,
and similarly $\hat\psi_\es^{\ft_2}$ is a function of
$(\hat B,\hat G, \hat\tau)$.  
(If $\b=\oo$ then $\tau,\hat\tau$ are redundant
due to the boundary conditions.)  Write 
\[
\cX=\cX_\b= \Big[\bigcup_{i=0}^d(\ZZ^d+\tfrac{1}{2}e_i)\Big]
\times I_\b,\qquad
T=\{0,1\}^{\ZZ^d},
\]
where $e_0$ is the zero vector and, for $i\neq0$,
$e_i$ is the unit vector in the $i$:th coordinate.
Writing $\cN=\cN_\cX$ for the set of boundedly
finite point processes 
(counting measures) on $\cX$ (denoted $\cN_\cX^{\#}$
in~\cite{dvj}), an obvious mapping
allows us to see both $B$ and $\hat B\cup\hat G$
as random elements of $\cN$.
The measure $\ol\PP_{N,r}$ factorizes as
\begin{equation}\label{P_decomp}
\overline\PP_{N,r}=\PP_{N,r}^{\ft_1}\times 
\hat\PP_{N,r}^{\ft_2}\times P_\D,
\end{equation}
where $P_\D$ is the law of $\D$ and 
$\PP_{N,r}^{\ft_1}$, $\hat\PP_{N,r}^{\ft_2}$
are the measures on $\cN\times T$ governing
$(B,\tau)$ and $(\hat B\cup \hat G,\hat\tau)$
respectively, given by
\[
\frac{d\PP_{N,r}^{\ft_1}}{dE_{N,r}}=
\frac{\partial\psi_\es^{\ft_1}}
{E_{N,r}(\partial\psi_\es^{\ft_1})},\qquad
\frac{d\hat\PP_{N,r}^{\ft_2}}{d\hat E_{N,r}}=
\frac{\partial\hat\psi_\es^{\ft_2}}
{\hat E_{N,r}(\partial\hat\psi_\es^{\ft_2})}.
\]
Write $\Om=(\cN\times T)^2\times\cN$ and note
that $\cN$, and hence also $\Om$, is a
complete and separable metric 
space~\cite[Proposition~9.1.IV]{dvj}.
We equip $\Om$ with the Borel $\s$-algebra
$\cF$, which
coincides with the $\s$-algebra generated by 
finite-dimensional distributions.
Recall that $r$ is either
fixed (if $\b<\oo$) or $r=2N$ (if $\b=\oo$).
\begin{proposition}\label{weak_prop}
The measures $\ol\PP_{N,r}$ converge 
weakly  to a probability measure $\ol\PP$
on $\Om$ as $N\to\oo$.  
\end{proposition}
\begin{proof}
The measure $P_\D$ does not depend on $N$, 
so by~\cite[Theorem~2.8]{billingsley-conv}
convergence of $\ol\PP_{N,r}$ follows once we 
show convergence of $\PP_{N,r}^{\ft_1}$
and $\hat\PP_{N,r}^{\ft_2}$.  We give full details for the
case of $\PP_{N,r}^{\ft_1}$, the case of 
$\hat\PP_{N,r}^{\ft_2}$ is similar.  We first show that
the sequence $\PP_{N,r}^{\ft_1}$ is tight, i.e.\ every
subsequence can be refined to a further subsequence 
along which $\PP_{N,r}^{\ft_1}$ converges.  We then show
that there is a $\pi$-system $\cA_0$ such that
the limit of $\PP_{N,r}^{\ft_1}(A)$ exists for each
$A\in\cA_0$.  The result then follows in 
the standard way from Prohorov's theorem.

Turning to the tightness of $\PP_{N,r}^{\ft_1}$,
note that since $T$ is compact it suffices
to show that the marginal of $\PP_{N,r}^{\ft_1}$
on $\cN$ is tight.  A criterion for this
is given in~\cite[Proposition~11.1.VI]{dvj}.
Fix $N_0,r_0$ and write $X=|B\cap F(N_0,r_0)|$
for the number of points of $B$ `inside'
$K(N_0,r_0)$.  Tightness follows if we show 
that for for each $\eps>0$ there is $m$
such that
\[
\PP_{N,r}^{\ft_1}(X>m)<\eps\mbox{ for all } N,r.
\]
By Markov's inequality it suffices to show that 
there is a constant $C(N_0,r_0)$ not depending
on $N,r$ such that the expectation
\begin{equation}\label{tight_ineq}
\PP^{\ft_1}_{N,r}[X] \leq C(N_0,r_0)
\mbox{ for all } N,r.
\end{equation}
In proving~\eqref{tight_ineq} we will require
the following notation.  Write 
$K=K(N,r)$, $K_0=K(N_0,r_0)$, $K'=K\sm K_0$, 
$F=F(N,r)$, $F_0=F(N_0,r_0)$, $F'=F\sm F_0$,
$B'=B\cap F'$, and $B^{(0)}=B\cap F_0$.
For briefer notation write $\psi$ for
the labelling $\psi_\es^{\ft_1}$, and let 
$\psi'$ denote the restriction of $\psi$
to $K'$.  Recall that $\epsilon$ denotes
the total Lebesgue measure of $K$ labelled
`even' in $\psi$, and write $\epsilon'$
for the  Lebesgue measure of the `even'
subset of $K'$.  Letting $\cC$ denote the event 
that $B$ is consistent (with source set $\es$) 
we have that $\partial\psi=e^{2\d\epsilon}\one_\cC$.
Clearly
$\epsilon'\leq \epsilon\leq \epsilon'+|K_0|$,
where $|K_0|=(2N_0+1)^dr_0$ denotes the total
Lebesgue measure of $K_0$, and it follows that
\begin{equation}\label{E_ratio}
\PP^{\ft_1}_{N,r}[X]\leq
e^{2\d |K_0|}\frac{E_{N,r}(X e^{2\d \epsilon'}\one_\cC)}
{E_{N,r}(e^{2\d \epsilon'}\one_\cC)}.
\end{equation}
The difficulty lies in the fact that although
$X$ is a function of $B^{(0)}$ only,
both $\epsilon'$ and $\cC$ depend both on both
$B^{(0)}$ and $B'$.  To `separate' this dependence,
we introduce the notation
\[
Z_{xy}=|B_{xy}\cap I_r|,\quad Z^{(0)}_{xy}=|B_{xy}\cap I_{r_0}|,
\quad Z'_{xy}=|B_{xy}\sm I_{r_0}|,\quad (xy\in\cE_N).
\]
Thus $Z_{xy}=Z^{(0)}_{xy}+Z'_{xy}$,
and the $Z^{(0)}_{xy},Z'_{xy}$ are independent Poisson
random variables under $E_{N,r}$
(with parameters $\l r_0$ and $\l(r-r_0)$, 
respectively).  The number  of switching points
on $\{x\}\times I_r$ can be written
\[
|S_x|=\sum_{\substack{y\in \L_N\\y\sim x}}Z_{xy},
\]
and letting $S^{(0)}_x$ denote the set of switching
points in $\{x\}\times I_{r_0}$ we can similarly write
\[
|S^{(0)}_x|=\sum_{\substack{y\in \L_{N_0}\\y\sim x}}Z^{(0)}_{xy}.
\]
The process $B'$ imposes parity constraints on
$B^{(0)}$ which can be described in terms of the
random vector $\pi=(\pi_x:x\in\L_{N_0})\in\{0,1\}^{\L_{N_0}}$
given by 
\begin{equation*}
\pi_x\equiv \sum_{\substack{y\not\in\L_{N_0}\\y\sim x}} Z^{(0)}_{xy}
+\sum_{\substack{y\in \L_N\\y\sim x}} Z'_{xy}.
\end{equation*}
(Here and in what follows we write $\equiv$ for
congruence modulo 2.)  Note that $\pi$ is a 
function of $B'$ only, and that
$\cC=\cC'\cap\cC^{(0)}\cap\tilde\cC$, where
\[
\begin{split}
\cC'&=\{|S_x|\equiv 0\;\forall x\in\L_N\sm\L_{N_0}\},\\
\cC^{(0)}&=\{|S_x^{(0)}|\equiv \pi_x\;\forall x\in\L_{N_0}\},\\
\tilde\cC&=
\Big\{\exists\, z\in\{0,1\}^{\cE_{N_0}}:\forall x\in\L_{N_0},
\sum_{\substack{y\sim x\\y\in\L_{N_0}}} z_{xy}\equiv \pi_x\Big\}.
\end{split}
\]
Strictly speaking the event $\tilde\cC$ is redundant
as it is implied by $\cC^{(0)}$, however
it is useful to keep since it, in contrast to $\cC^{(0)}$,
depends on $B'$ only.  For each realization of $\pi$
such that $\tilde\cC$ holds we fix a deterministic
vector $z$ as in the definition of $\tilde\cC$.
Note that the number of possible $\pi$ is at most
$2^{|\L_{N_0}|}$.

In the numerator of~\eqref{E_ratio} we have
\[\begin{split}
E_{N,r}(Xe^{2\d\epsilon'}\one_{\cC})&=
E_{N,r}\big(e^{2\d\epsilon'}\one_{\cC'\cap\tilde\cC}
E_{N,r}(X\one_{\cC^{(0)}}\mid\psi')\big)\\
&\leq E_{N,r}(X)E_{N,r}(e^{2\d\epsilon'}\one_{\cC'\cap\tilde\cC}),
\end{split}\]
where we bounded $\one_{\cC^{(0)}}$ by 1 and used the fact
that $X$ is independent of $\psi'$.
Note that $E_{N,r}(X)=\l r_0|\cE_{N_0}|$.
In the denominator of~\eqref{E_ratio} we have
\[
E_{N,r}(e^{2\d\epsilon'}\one_{\cC})=
E_{N,r}\big(e^{2\d\epsilon'}\one_{\cC'\cap\tilde\cC}
P_{N,r}(\cC^{(0)}\mid \psi')\big),
\]
which we need to bound from below.
We claim that there is an $\eps=\eps(N_0,r_0)>0$
such that $P_{N,r}(\cC^{(0)}\mid \psi')\geq \eps$
for all realizations $\psi'$ such that $\tilde\cC$
holds.  Indeed, recall that we fixed a deterministic
vector $z$ for each $\pi$ such that $\tilde\cC$ holds.
The event 
$\tilde\cC_z=\{Z^{(0)}_{xy}=z_{xy}\forall xy\in \cE_{N_0}\}$
thus implies $\cC^{(0)}$.   Under 
$P_{N,r}(\cdot\mid \psi')$ the $Z^{(0)}_{xy}$
are independent Poisson random variables,
so each $\tilde\cC_z$ has positive probability.
The claim thus holds with $\eps$
being the minimum of $P_{N,r}(\tilde\cC_z\mid \psi')$
over the (at most $2^{|\L_{N_0}|}$) 
choices of $z$.
Therefore~\eqref{tight_ineq} follows with
\[
C(N_0,r_0)=e^{2\d |K_0|}\l r_0|\cE_{N_0}|/\eps(N_0,r_0).
\]

Having proved tightness of the sequence $\PP_{N,r}^{\ft_1}$
we now turn to showing uniqueness of subsequential
limits.  Let $\cA_0$ denote the
collection of events of the form 
\begin{equation*}
A=\{\psi\mbox{ is `even' in } J\},
\end{equation*}
where $J$ is any finite union of bounded 
closed intervals
in $\KK_\b$.   
(We allow intervals of length 0, i.e.\ isolated
points.)
Then $\cA_0$ is a  $\pi$-system which
generates the $\s$-algebra $\cF$
(note that the process $B$ can be recovered 
from the labelling $\psi$).
We let $N$ be large enough that $J\se K(N,r)$.
By~\eqref{holes_eq} and~\cite[Lemma~3.2]{bjogr}
we have  
\begin{equation}\label{LJ_eq}
  \PP^{\ft_1}_{N,r}(A)=c(J)\mu_{N,r}^{\f,\ft_1}
\big[\exp\big(-\l L_J(\s)\big)\big],
\end{equation}
for some constant $c(J)$ depending only on $J$,
and where
\begin{equation*}
L_J(\s)=\sum_{xy\in \cE_N}
\int_{I_r}\s(x,t)\s(y,t)
\one\{(xy,t)\in \tilde J\}dt
\end{equation*}
and $\tilde J$ is the
set of point $(yz,t)\in\FF$ such that 
$(y,t)\in J$ or $(z,t)\in J$ (or both).
By~\cite[Theorem~2.5.1]{bjo_phd} 
the sequence of measures $\mu^{\f,\ft_1}_{N,r}$
converges weakly, hence the probability 
in~\eqref{LJ_eq} converges.
\end{proof}

For later reference we note that the constant $c(J)$
in~\eqref{LJ_eq} can be written as
\begin{equation}\label{cJ_eq}
c(J)=2^{-n}e^{\d|J|+\l|\tilde J|}.
\end{equation}
Here $|J|$ and $|\tilde J|$ denote the total length
of the intervals comprising $J$ and $\tilde J$,
and $n$ is the difference between the number
of intervals comprising $K'=K\sm J$ and 
the number of intervals comprising $K$
(in counting the number of intervals we view
$I_\b$ as a circle when $\ft=\p$).

For $(x,s)\in\KK_\b$ define the translation
or shift $\tau_{(x,s)}:\KK_\b\to\KK_\b$ by
$\tau_{(x,s)}(y,t)=(y+x,t+s)$
where in the case $\b<\oo$ we view $t+s$
modulo $\b$.  For a function $\zeta:\KK_\b\to\RR$
(e.g.\ a labelling $\psi$
or a spin-configuration $\s$) we define 
$\tau_{(x,s)}(\zeta)$
by $[\tau_{(x,s)}(\zeta)](y,t)=\zeta(y+x,t+s)$.
We write $\tau_x$ for $\tau_{(x,0)}$.
\begin{proposition}\label{ergodic_prop}
The measure $\ol\PP$ is invariant with respect
to the shifts $\tau_{(x,s)}$, and
ergodic with respect to the shifts $\tau_x$
for $x\neq0$.
\end{proposition}
\begin{proof}
We use the decomposition~\eqref{P_decomp}.
The measure $P_\D$ is translation-invariant
and ergodic, so it suffices to show that the weak 
limits of $\PP_{N,r}^{\ft_1}$ and  
$\hat\PP_{N,r}^{\ft_2}$ and are translation-invariant 
and ergodic.  Again, we give details for
$\PP_{N,r}^{\ft_1}$.  Let $A\in\cA_0$,
let $(x,s)\in\KK_\b$, and let $N$ be large enough 
that $\tau_{(x,s)}^{-1}J\se K(N,r)$
(recall that $r=2N$ if $\b=\oo$ and $r=\b$
if $\b<\oo$).  We have as in~\eqref{LJ_eq} that
\begin{equation*}
  \PP^{\ft_1}_{N,r}(\tau_{(x,s)}A)=c(\tau_{(x,s)}^{-1}J)
\mu_{N,r}^{\f,\ft_1}
\big[\exp\big(-\l L_{\tau_{(x,s)}^{-1}J}(\s)\big)\big].
\end{equation*}
From~\eqref{cJ_eq} we see that
 $c(\tau_{(x,s)}^{-1}J)=c(J)$.
By~\cite[Theorem~2.5.1]{bjo_phd},
\[\begin{split}
\lim_{N\to\oo} \mu_{N,r}^{\f,\ft_1}
\big[\exp\big(-\l L_{\tau_{(x,s)}^{-1}J}(\s)\big)\big]
&=\lim_{N\to\oo} \mu_{N,r}^{\f,\ft_1}
\big[\exp\big(-\l L_J(\s)\big)\big]\\
&=\mu^{(\f,\ft_1)}
\big[\exp\big(-\l L_J(\s)\big)\big],
\end{split}\]
and hence
$\lim_{N\to\oo}\PP^{\ft_1}_{N,r}(\tau_{(x,s)}A)=
\lim_{N\to\oo}\PP^{\ft_1}_{N,r}(A)$.
This proves translation-invariance on the
$\pi$-system $\cA_0$, which by the $\pi$-systems
lemma implies full
translation-invariance.

Let $A_1,A_2\in\cA_0$ be the events that
$\psi$ is `even' in $J_1$ and $J_2$, respectively.
For $\|x\|$ and $N$ large enough we have 
from~\eqref{LJ_eq} and~\eqref{cJ_eq} that
\[
\frac{\PP_{N,r}^{\ft_1}(A_1\cap\tau_x A_2)}
{\PP_{N,r}^{\ft_1}(A_1)\PP_{N,r}^{\ft_1}(A_2)}
=\frac{\mu_{N,r}^{\f,\ft_1}
\big[\exp\big(-\l L_{J_1}(\s)\big)
\exp\big(-\l L_{\tau^{-1}_{x}J_2}(\s)\big)\big]}
{\mu_{N,r}^{\f,\ft_1}
\big[\exp\big(-\l L_{J_1}(\s)\big)\big]
\mu_{N,r}^{\f,\ft_1}
\big[\exp\big(-\l L_{\tau^{-1}_{x}J_2}(\s)\big)\big]}.
\]
Letting $N\to\oo$ and writing $\PP^{(\ft_1)}$ for
the weak limit of $\PP_{N,r}^{\ft_1}$ it follows that
\[
\frac{\PP^{(\ft_1)}(A_1\cap\tau_x A_2)}
{\PP^{(\ft_1)}(A_1)\PP^{(\ft_1)}(A_2)}
=\frac{\mu^{(\f,\ft_1)}
\big[\exp\big(-\l L_{J_1}(\s)\big)
\exp\big(-\l L_{J_2}(\tau^{-1}_{x}\s)\big)\big]}
{\mu^{(\f,\ft_1)}
\big[\exp\big(-\l L_{J_1}(\s)\big)\big]
\mu^{(\f,\ft_1)}
\big[\exp\big(-\l L_{J_2}(\s)\big)\big]}.
\]
It follows from Lemma~\ref{mix_lem} 
in the Appendix that the
ratio on the right-hand-side converges to 1
as $\|x\|\to\oo$.  Thus $\PP^{(\ft_1)}$
is mixing on $\cA_0$, hence also mixing on $\cF$ 
and hence ergodic.
\end{proof}

\subsection{Percolation}
\label{perc_sec}

The notions of paths and connectivity, defined 
for $\ol\PP_{N,r}$, extend to $\ol\PP$.
Thus $\KK_\b$ decomposes into a random 
collection of connected components or clusters
(each of these is a union of intervals
bounded by certain elements of $\D$).  Let $U$ denote the 
number of these clusters which are \emph{unbounded}.
The random variable $U$ (which may be infinite)
is invariant under all translations $\tau_x$,
and hence by Proposition~\ref{ergodic_prop} it is
$\ol\PP$-a.s.\ constant.  We will show:

\begin{proposition}\label{perc_prop}
Either $\ol\PP(U=0)=1$ or $\ol\PP(U=1)=1$.
\end{proposition}
\begin{proof}
As this type of argument is fairly standard in 
percolation theory,
and most of the details are the same as for the classical
model~\cite{adcs}, we only sketch the proof and highlight
what adjustments are needed for the quantum case.
We focus on the case $\b=\oo$.  

Let $k$ be such that $\ol\PP(U=k)=1$, we must show that
$k\leq1$.  The possibility that $2\leq k<\oo$
may be ruled out using 
Proposition~\ref{local_prop}.  Roughly speaking,
a large enough box $K(N,r)$ will intersect all
$k$ unbounded components with positive probability.
Using part (B) of Proposition~\ref{local_prop}
one may deduce that, with positive probability,
all these components are in fact connected
inside $K(N,r)$, a contradiction.

Now assume that $k=\oo$.  To get a contradiction
one considers what are called 
\emph{coarse-trifurcations}.  These are
points $(x,t)\in\KK$ with
$x\in (2N_0+1)\ZZ^d$
and $t\in 2r_0\ZZ$, having the properties that
(i) all points in $K_0+(x,t)$
are connected in $K_0+(x,t)$, and 
(ii) $\KK\sm (K_0+(x,t))$
contains at least 3 distinct unbounded connected
components.  Here $N_0$ and $r_0$ are fixed,
and may (using Proposition~\ref{local_prop}(B)
again) be chosen such that each $(x,t)$
is a coarse-trifurcation with probability $p>0$.
Note that $p$ is the same for all $(x,t)$
by translation-invariance.  As in~\cite{adcs}
one may construct a graph $F$ which reflects the
connectivity structure of the coarse-trifurcations
in some large box $K(N,r)$.  Roughly speaking,
the edges of $F$ either connect distinct
coarse-trifurcations, or they connect a coarse-trifurcation 
with an `element' on the boundary of $K(N,r)$.
The latter occurs if the 
boundary can be reached without having to
go too close to another coarse-trifurcation,
otherwise the former occurs.

The main difference to~\cite{adcs} is the correct
notion of `element' on the boundary.  In their
case, where the underlying graphical structure is
discrete, one may take the `elements' simply as vertices
on the boundary.  In the present case we instead 
take it to mean a  maximal $\D$-free interval
intersecting the boundary of $K(N,r)$;  that is an
interval of the form $\{y\}\times I$ with $I\se I_r$ 
maximal such that $\D_y\cap I=\es$,  and such that either  
$y\in\partial\L_N$, or one of the endpoints of $I$
is $\pm r/2$.  With this convention the graph $F$
contains no cycles, as if it did contain a cycle
this would violate the definition of coarse-trifurcation.  
(The important point is that the
collection of maximal $\D$-free
intervals refines the collection of connected components.)
Also, the coarse-trifurcations correspond to vertices
of $F$ of degree at least 3.  This implies that
the number of coarse-trifurcations is at most the
number of leaves of $F$, which is in turn bounded above
by the number of maximal $\D$-free interval
intersecting the boundary of $K(N,r)$.  The expectation 
of the latter is easily seen to be at most
$2(2N+1)^d+4\d r(2N+1)^{d-1}$, whereas the
expected number of coarse-trifurcations is of the
order $p[r(2N+1)^d]/[r_0(2N_0+1)^d]$.  
Letting $N,r\to\oo$ this contradicts $p>0$,
finishing the proof.
\end{proof}

\section{Proof of the main result}
\label{proof_sec}

\subsection{The infrared bound}
We now describe the infrared bound of~\cite{bjo_irb}
and use it to prove a result of key importance
for Theorem~\ref{main_thm}.  For technical
reasons we will in this section redefine the box 
$\L_n$ as $\{-n+1,\dotsc,n\}^d$ so that it has even
sidelength rather than odd.
Recall the Schwinger function~\eqref{schwinger_eq}
and its probabilistic representation~\eqref{corr_eq2}.
Write
\begin{equation*}
c_{N,r}(x,t)=\el\s(0,0)\s(x,t)\er_{N,r}^{\p,\p}.
\end{equation*}
Note that we use periodic boundary conditions
in both directions.
Although $c_{N,r}(x,t)$ is defined for $(x,t)\in K(N,r)$,
we extend the definition
to all of $\KK$ by periodicity.  We write
\begin{equation*}
\L^\star_N=\tfrac{\pi}{N}\L_N,\quad
I_r^\star=\tfrac{2\pi}{r}\ZZ,\quad
K^\star_{N,r}=\L^\star_N\times I_r^\star.
\end{equation*}
Elements of $K^\star_{N,r}$ will be denoted $\xi=(k,\ell)$
where $k\in \L_N^\star$ and $\ell\in I_r^\star$.
For large $N,r$ we may see $\L_N^\star$ 
as an approximation of $(-\pi,\pi]^d$
and $I_r^\star$ as an 
approximation of $\RR$.  
For $p=(p_1,\dotsc,p_d)\in(-\pi,\pi]^d$ let 
$\hat L(p)=\sum_{j=1}^d(1-\cos(p_j))$
denote the Fourier transform of the graph Laplacian of $\ZZ^d$,
and define 
\begin{equation*}
E_{\l,\d}(p,q)=\frac{2\l\hat L(p)+q^2/2\d}{48},
\qquad p\in(-\pi,\pi]^d, q\in\RR.
\end{equation*}
The Fourier transform of $c_{N,r}$ is 
\begin{equation*}
\hat c_{N,r}(\xi)=\sum_{x\in \L_N}\int_{I_r} c_{N,r}(x,t)
e^{ik\cdot x} e^{i\ell t}\,dt,\quad
\xi=(k,\ell)\in  K_{N,r}^\star,
\end{equation*}
where $k\cdot x$ denotes the usual scalar product.
Note that $\hat c_{N,r}(\xi)\geq0$.
The infrared bound of~\cite{bjo_irb} states that
\begin{equation}\label{irb}
\mbox{if } \xi\in K_{N,r}^\star\sm\{0\}
\mbox{ then } \hat c_{N,r}(\xi)\leq \frac{1}{E_{\l,\d}(\xi)}.
\end{equation}
We will use this to show the following:
\begin{lemma}\label{avg_lem}
Suppose $\b<\oo$ and $d\geq 3$.  Then 
\begin{equation}\label{avg_pos}
\lim_{n\to\oo}\frac{1}{|\L_n|}\sum_{x\in \L_n}\int_{I_\b}
\el\s(0,0)\s(x,t)\er_{\l_\crit,\b}^{(\f,\p)}\,dt=0.
\end{equation}
Suppose $\b=\oo$ and $d\geq 2$.  Then 
\begin{equation}\label{avg_ground}
\lim_{n,r\to\oo}\frac{1}{|\L_n|r}\sum_{x\in \L_n}\int_{I_r}
\el\s(0,0)\s(x,t)\er_{\l_\crit,\oo}^{(\f,\f)}\,dt=0.
\end{equation}
\end{lemma}
\begin{proof}
We begin by showing that the stated conditions
on $d$ imply the following:
\begin{equation}\label{L1_pos}
\mbox{if $\b<\oo$ then }
\int_{(-\pi,\pi]^d}dp\sum_{\ell\in I_\b^\star} \frac{1}{E_{\l,\d}(p,\ell)}<\oo,
\end{equation}
and
\begin{equation}\label{L1_ground}
\mbox{if $\b=\oo$ then }
\int_{(-\pi,\pi]^d}dp\int_\RR dq \frac{1}{E_{\l,\d}(p,q)}<\oo.
\end{equation}
Firstly, in the case $\b<\oo$ we have that
\begin{equation*}\begin{split}
\int_{(-\pi,\pi]^d}dp\sum_{\ell\in I_\b^\star}
\frac{1}{E_{\l,\d}(p,\ell)}
&=\int_{(-\pi,\pi]^d}dp\sum_{\ell\in I_\b^\star}
\frac{48}{2\l\hat L(k)+\ell^2/2\d}\\
&=\frac{96\d}{(2\pi/\b)^2}\int_{(-\pi,\pi]^d}dp
\sum_{m\in\ZZ}\frac{1}{a(p)+m^2}\\
&\leq\frac{96\d}{(2\pi/\b)^2}\int_{(-\pi,\pi]^d}
\Big(\frac{1}{a(p)}+\frac{\pi}{a(p)^{1/2}}\Big)dp,
\end{split}
\end{equation*}
where $a(p)=\l\d\b^2\hat L(p)/\pi^2$.  Note that 
$1/\hat L(p)$ diverges for $p\to 0$, in the manner
of $1/\|p\|_2^2$, and that for $\a>0$
\begin{equation}\label{L-int-alpha}
\int_{(-\pi,\pi]^d}\frac{1}{\hat L(p)^\a}dp<\oo
\mbox{ if and only if } d>2\a.
\end{equation}
Thus~\eqref{L1_pos} holds for $d>2$, i.e.\ for $d\geq3$
as claimed.  In the case $\b=\oo$ we have
\begin{equation*}\begin{split}
\int_{(-\pi,\pi]^d}dp\int_{-\oo}^\oo dq \frac{1}{E_{\l,\d}(p,q)}
&=96\d \int_{(-\pi,\pi]^d}dp\int_{-\oo}^\oo dq 
\frac{1}{4\l\d\hat L(p)+q^2}\\
&=48\pi\sqrt{\d/\l}
\int_{(-\pi,\pi]^d}\frac{1}{\hat L(p)^{1/2}}dp.
\end{split}\end{equation*}
Thus by~\eqref{L-int-alpha} we have~\eqref{L1_ground}
for $d>1$, i.e.\ for $d\geq 2$ as claimed.

Now define the function
\begin{equation*}
G_{N,r}((x,s),(y,t))=\frac{1}{(2N)^dr}
\sum_{(k,\ell)\in K_N^\star\sm\{0\}}
\frac{e^{-ik\cdot (x-y)}e^{-i\ell(s-t)}}{E_{\l,\d}(k,\ell)},
\end{equation*}
where $x,y\in\ZZ^d$ and $s,t\in\RR$.  In the case when
$\b<\oo$ then by Riemann approximation
\begin{equation*}
\begin{split}
G_{N,\b}((x,s),(y,t))&\to\frac{1}{(2\pi)^d\b}
\int_{(-\pi,\pi]^d}dp\sum_{\ell\in I_\b^\star} 
\frac{e^{-ip\cdot (x-y)}e^{-i\ell(s-t)}}{E_{\l,\d}(p,\ell)}\\
&=:G_{\b}((x,s),(y,t)),\mbox{ as }N\to\oo.
\end{split}
\end{equation*}
From~\eqref{L1_pos}  we deduce that
\[
\frac{1}{(2n)^{d}}\sum_{x\in\L_n} \int_{I_\b}
dt \,G_\b((0,0),(x,t))\to 0\mbox{ as }n\to\oo,
\]
which in turn implies that 
\begin{equation}\label{RL_pos}
\frac{1}{(2n)^{2d}}\sum_{x,y\in\L_n} \iint_{I_\b\times I_\b}
dsdt \,G_\b((x,s),(y,t))\to 0\mbox{ as }n\to\oo.
\end{equation}
Similarly, for $\b=\oo$ we have that
\begin{equation*}\begin{split}
G_{N,r}((x,s),(y,t))&\to\frac{1}{(2\pi)^{d+1}}
\int_{(-\pi,\pi]^d}dp\int_\RR dq
\frac{e^{-ip\cdot (x-y)}e^{-iq(s-t)}}{E_{\l,\d}(p,q)}\\
&=:G_{\oo}((x,s),(y,t)),\mbox{ as }N,r\to\oo,
\end{split}
\end{equation*}
and hence using~\eqref{L1_ground} that 
\begin{equation}\label{RL_ground}
\frac{1}{(2n)^{2d}r^2}\sum_{x,y\in\L_n} \iint_{I_r\times I_r}
dsdt \,G_\oo((x,s),(y,t))\to 0\mbox{ as }n,r\to\oo.
\end{equation}

We now show how~\eqref{RL_pos} and~\eqref{RL_ground}
imply~\eqref{avg_pos} and~\eqref{avg_ground}, respectively.
Note that by Fourier inversion
\begin{equation*}
c_{N,r}(x,t)=\frac{1}{(2N)^dr}\sum_{k\in\L_N^\star}
\sum_{\ell\in I_r^\star} \hat c_{N,r}(k,\ell) e^{-ik\cdot x}e^{-i\ell t}.
\end{equation*}
Let $v:K(N,r)\to\CC$ be an aribtrary bounded, measurable function.  
It follows that
\begin{equation}\label{v_eq}
\begin{split}
\sum_{x,y\in\L_N}&\iint_{I_r\times I_r}dsdt \,v(x,s)\overline{v(y,t)}
c_{N,r}(x-y,s-t)\\
&=\frac{1}{(2N)^dr}\sum_{\xi\in K_{N,r}^\star}
\hat c_{N,r}(\xi) |z_v(\xi)|^2,
\end{split}
\end{equation}
where 
\begin{equation*}
z_v(k,\ell)=\sum_{x\in\L_N}\int_{I_r}v(x,s) e^{-ik\cdot x}e^{-i\ell s}ds.
\end{equation*}
Using the infrared bound~\eqref{irb},
\begin{equation*}
\sum_{\xi\in K_{N,r}^\star}
\hat c_{N,r}(\xi) |z_v(\xi)|^2
\leq\sum_{\xi\in K_{N,r}^\star\sm\{0\}}
\frac{1}{E_{\l,\d}(\xi)} |z_v(\xi)|^2+
\hat c_{N,r}(0)  |z_v(0)|^2.
\end{equation*}
Note that
\begin{equation*}
\hat c_{N,r}(0)  =\sum_{x\in\L_N}\int_{I_r} c_{N,r}(x,t)dt=:
\chi^{\p,\p}_{N,r}
\end{equation*}
equals the (finite-volume)  susceptibility.
Interchanging the order of summation again thus gives
\begin{multline}\label{v_rhs_eq}
\frac{1}{(2N)^dr}\sum_{\xi\in K_{N,r}^\star}
\hat c_{N,r}(\xi) |z_v(\xi)|^2\\\leq
\sum_{x,y\in\L_N}\iint_{I_r\times I_r} v(x,s) \overline{v(y,t)}
G_{N,r}((x,s),(y,t))+\frac{|z_v(0)|^2}{(2N)^dr} \chi^{\p,\p}_{N,r}.
\end{multline}

Let $N_0<N$, and as usual let $r_0<r$ 
 if $\b=\oo$, alternatively $r_0=r=\b$
if $\b<\oo$.  Set
\[
v(x,s)=\one\{x\in\L_{N_0},s\in I_{r_0}\}.
\]
In what follows we use the same notation 
$\el\cdot\er_{N,r}^{\f}$ for both
$\el\cdot\er_{N,r}^{\f,\f}$ (in the case
$\b=\oo$) and $\el\cdot\er_{N,\b}^{\f,\p}$ (in the case
$\b<\oo$).  We also write
$\el\cdot\er_{\l,\b}^{(\f)}$ for both
infinite-volume limits $\el\cdot\er_{\l,\oo}^{(\f,\f)}$
and $\el\cdot\er_{\l,\b}^{(\f,\p)}$.  By the monotonicity~\eqref{corr_comp} 
of  correlation functions we have that
\begin{equation*}
c_{N,r}(x-y,s-t)=\el\s(x,s)\s(y,t)\er_{N,r}^{\p,\p}
\geq \el\s(x,s)\s(y,t)\er_{N,r}^{\f}.
\end{equation*}
Thus for our choice of $v$ we have that
the left-hand-side of~\eqref{v_eq} is  at
least 
\begin{equation*}
\sum_{x,y\in\L_{N_0}}\iint_{I_{r_0}\times I_{r_0}}
\el\s(x,s)\s(y,t)\er_{N,r}^{\f} dsdt.
\end{equation*}
By~\eqref{v_rhs_eq} it follows that 
\begin{equation}\label{chi_bound}
\begin{split}
\sum_{x,y\in\L_{N_0}}&\iint_{I_{r_0}\times I_{r_0}}
\el\s(x,s)\s(y,t)\er_{N,r}^{\f} dsdt\\
&\leq \sum_{x,y\in\L_{N_0}}\iint_{I_{r_0}\times I_{r_0}} dsdt
G_{N,r}((x,s),(y,t))+\frac{(2N_0)^dr_0}{(2N)^dr} \chi^{\p,\p}_{N,r}.
\end{split}
\end{equation}

Now let $\l<\l_\crit$.  This 
implies that $\el\cdot\er_{\l,\b}^{(\f)}$ is the unique
infinite-volume limit of the
measures $\el\cdot\er_{N,r}^{\fs,\ft}$, and
by finiteness of the susceptibility~\cite[Theorem~6.6]{bjogr} 
and the dominated convergence theorem we have that 
\begin{equation*}
\chi^{\p,\p}_{N,r}\to 
\sum_{x\in\ZZ^d}\int_{I_\b}\el\s(0,0)\s(x,t)\er_{\l,\b}^{(\f)}dt<\oo,
\end{equation*}
as $N\to\oo$ (for $\b=r<\oo$) or 
$N,r=2N\to\oo$ (for $\b=\oo$).
Hence, letting $N\to\oo$ or $N,r\to\oo$ as appropriate, we obtain
from~\eqref{chi_bound}  that 
for all $\l<\l_\crit$ we have
\begin{equation}\label{G_bound}
\begin{split}
\sum_{x,y\in\L_{N_0}}&\iint_{I_{r_0}\times I_{r_0}}
\el\s(x,s)\s(y,t)\er_{\l,\b}^{(\f)} dsdt\\
&\leq \sum_{x,y\in\L_{N_0}}\iint_{I_{r_0}\times I_{r_0}} 
G_\b((x,s),(y,t))\,dsdt.
\end{split}
\end{equation}
Letting $\l\uparrow\l_\crit$ and using the fact that 
$\el\s(x,s)\s(y,t)\er_{\l,\b}^{(\f)}$ is left-continuous in $\l$
(since any two increasing limits can be interchanged),
we get that~\eqref{G_bound} holds also with $\l=\l_\crit$.
Letting 
$N_0\to\oo$ or $N_0,r_0\to\oo$ as
appropriate we get from~\eqref{RL_pos} and~\eqref{RL_ground} 
that
\[
\frac{1}{(2N_0)^{2d}r_0^2}
\sum_{x,y\in\L_{N_0}}\iint_{I_{r_0}\times I_{r_0}}
\el\s(x,s)\s(y,t)\er_{\l_\crit,\b}^{(\f)} \,dsdt\to0.
\]
Using translation-invariance and nonnegativity
of $\el\s(x,s)\s(y,t)\er_{\l_\crit,\b}^{(\f)}$, the 
results~\eqref{avg_pos} and~\eqref{avg_ground} follow.
\end{proof}

\subsection{Proof of Theorem~\ref{main_thm}}

From this point the argument is almost identical
to that for the classical model~\cite{adcs},
however it is also short and elegant so we include 
the remaining steps.
We begin by deducing from Lemma~\ref{avg_lem} 
the following consequence for the number
$U$ of unbounded components under the measure $\ol\PP$.

\begin{proposition}\label{crit-perc_prop}
Under the conditions in Lemma~\ref{avg_lem}
and for $\l=\l_\crit$
we have that $\ol\PP(U=0)=1$.
\end{proposition}
\begin{proof} 
Recall our convention on boundary conditions for the
measure $\ol\PP_{N,r}$:  if $\b<\oo$ we write $\ft_1=\ft_2=\p$,
if $\b=\oo$ we write $\ft_1=\f$ and $\ft_2=\w$.  By the
definition of $\ol\PP_{N,r}$ and the Switching 
Lemma~\ref{sw_lem} we have for any $N,r$
and $(x,s),(y,t)\in K(N,r)$ that
\[\begin{split}
\ol\PP_{N,r}((x,s)\lra(y,t))&=
\el\s(x,s)\s(y,t)\er_{N,r}^{\f,\ft_1}
\el\s(x,s)\s(y,t)\er_{N,r}^{\w,\ft_2}\\
&\leq\el\s(x,s)\s(y,t)\er_{N,r}^{\f,\ft_1}.
\end{split}\]
By the convergence of the correlation function
and using Proposition~\ref{weak_prop}
(and a small, but standard, additional argument)
we get that
\begin{equation}\label{p-c}
\ol\PP((x,s)\lra(y,t))\leq 
\el\s(x,s)\s(y,t)\er_{\l_\crit,\b}^{(\f,\ft_1)}.
\end{equation}
Writing $\{(x,s)\lra\oo\}$ for the event that
$(x,s)$ lies in an unbounded component we have
using Jensen's inequality and the fact that
$U\leq1$ (Proposition~\ref{perc_prop}) that
\[\begin{split}
\Big( r|\L_N| \ol\PP((0,0)\lra\oo)\Big)^2
&\leq \ol\PP\Big( \Big(\sum_{x\in\L_N} 
\int_{I_r} \one\{(x,s)\lra\oo\}ds\Big)^2\Big)\\
&= \sum_{x,y\in\L_N} 
\iint_{I_r\times I_r} \ol\PP((x,s),(y,t)\lra\oo)dsdt\\
&\leq \sum_{x,y\in\L_N} 
\iint_{I_r\times I_r} \ol\PP((x,s)\lra(y,t))dsdt\\
&= r|\L_N|\sum_{x\in\L_N} 
\int_{I_r} \ol\PP((0,0)\lra(x,s))ds.
\end{split}\]
Using~\eqref{p-c} we deduce that
\[
\ol\PP((0,0)\lra\oo)^2\leq
\frac{1}{r|\L_N|}\sum_{x\in\L_N} 
\int_{I_r} \el\s(0,0)\s(x,s)\er_{\l_\crit,\b}^{(\f,\ft_1)}
ds.
\]
Letting $N\to\oo$ or $N,r\to\oo$ as appropriate,
and using Lemma~\ref{avg_lem}, the result follows.
\end{proof}

Turning to the final steps in the 
proof of Theorem~\ref{main_thm}, we recall from
Proposition~\ref{local_prop} that for each 
$(x,t)\in\KK_\b$ there is a constant $C_{(x,t)}$
such that for all $N,r$ we have
\[
\el\s(0,0)\s(x,t)\er_{N,r}^{\w,\ft_2}
-\el\s(0,0)\s(x,t)\er_{N,r}^{\f,\ft_1}\leq
C_{(x,t)} \ol\PP_{N,r}((0,0)\lra\G).
\]
Also note that for all $N_0\leq N$ and $r_0\leq r$
we have that 
$\ol\PP_{N,r}((0,0)\lra\G)\leq\ol\PP_{N,r}((0,0)\lra\partial K(N_0,r_0))$,
since any path to $\G$ must leave $K(N_0,r_0)$.
Letting $N\to\oo$ (respectively, $N,r\to\oo$)
and then $N_0\to\oo$ (respectively, $N_0,r_0\to\oo$)
it follows from Proposition~\ref{crit-perc_prop} that 
\[
\el\s(0,0)\s(x,t)\er_{\l_\crit,\b}^{(\w,\ft_2)}
-\el\s(0,0)\s(x,t)\er_{\l_\crit,\b}^{(\f,\ft_1)}\leq
C_{(x,t)} \ol\PP((0,0)\lra\oo)=0.
\]
Thus $\el\s(0,0)\s(x,t)\er_{\l_\crit,\b}^{(\w,\ft_2)}
=\el\s(0,0)\s(x,t)\er_{\l_\crit,\b}^{(\f,\ft_1)}$.
By translation-invariance and 
the Griffiths inequality
(proved in detail for the present model 
in~\cite[Lemma~2.2.20]{bjo_phd}) it follows that
\[\begin{split}
\big(\el\s(0,0)\er_{\l_\crit,\b}^{(\w,\ft_2)}\big)^2&=
\el\s(0,0)\er_{\l_\crit,\b}^{(\w,\ft_2)}
\el\s(x,t)\er_{\l_\crit,\b}^{(\w,\ft_2)}\leq
\el\s(0,0)\s(x,t)\er_{\l_\crit,\b}^{(\w,\ft_2)}\\
&=\el\s(0,0)\s(x,t)\er_{\l_\crit,\b}^{(\f,\ft_1)}.
\end{split}\]
Thus using Lemma~\ref{avg_lem} again,
if $\b<\oo$ and $d\geq3$ or 
$\b=\oo$ and $d\geq2$ then 
\[
\big(\el\s(0,0)\er_{\l_\crit,\b}^{(\w,\ft_2)}\big)^2
\leq \frac{1}{r|\L_N|}\sum_{x\in\L_N}\int_{I_r}
\el\s(0,0)\s(x,t)\er_{\l_\crit,\b}^{(\f,\ft_2)}dt\to 0,
\]
hence by~\eqref{spont_resid} we have 
$M^+_\b(\l_\crit)=\el\s(0,0)\er_{\l_\crit,\b}^{(\w,\ft_2)}=0$
as claimed.\qed

\section{Appendix:  mixing in the space--time spin representation}

In this section we prove mixing results for
the infinite-volume space--time spin measures
$\mu^{(\fs,\ft)}_\b$ defined in Section~\ref{spin_sec}.
As usual we let $\ft=\p$ if $\b<\oo$ and
$\ft\in\{\f,\w\}$ if $\b=\oo$. 
As a shorthand we write
\[
\mu^{(\w)}=\left\{\begin{array}{ll}
\mu^{(\w,\w)}_\oo & \mbox{if } \b=\oo,\\
\mu^{(\w,\p)}_\b & \mbox{if } \b<\oo,
\end{array}\right.\qquad
\mu^{(\f)}=\left\{\begin{array}{ll}
\mu^{(\f,\f)}_\oo & \mbox{if } \b=\oo,\\
\mu^{(\f,\p)}_\b & \mbox{if } \b<\oo,
\end{array}\right.
\]
and $\el\cdot\er^{(\w)}$, $\el\cdot\er^{(\f)}$ for
the corresponding expectation operators.
For simplicity of presentation we focus on
the case $\b=\oo$, similar results and
constructions hold for the case $\b<\oo$.

To state and prove our mixing results we need 
to be precise about the topological set-up.
We define a metric $d$ on $\S$ as follows.
Firstly, for each $n\geq1$ define a `local' metric
\[
d_n(\s,\s')=\sum_{x\in\L_n}\int_{I_n} |\s(x,t)-\s'(x,t)|dt,
\qquad \s,\s'\in\S,
\]
and then extend this in a standard way by letting
\[
d(\s,\s') =\sum_{n\geq0}2^{-n}
\frac{d_n(\s,\s')}{1+d_n(\s,\s')}.
\]
Recall that a function $F:\S\to\RR$ is 
\begin{list}{$\bullet$}{\leftmargin=1em}
\item \emph{uniformly continuous} if for each $\eps>0$ there
is $\d>0$ such that if $d(\s,\s')<\d$ then
$|F(\s)-F(\s')|<\eps$;
\item \emph{even} if $F(\s)=F(-\s)$ for all $\s\in\S$.
\end{list}
We will  prove the following:
\begin{lemma}\label{mix_lem}
Let $0<\b\leq\oo$ and let $C_1,C_2:\S\to\RR$ be bounded,
uniformly continuous functions.  Then
\[
\lim_{\|x\|\to\oo}\el C_1(\s) [C_2\circ\tau_x](\s)]\er^{(\w)}
=\el C_1(\s)\er^{(\w)}\el C_2(\s)\er^{(\w)}.
\]
If, in addition, $C_1,C_2$ are even then also
\[
\lim_{\|x\|\to\oo}\el C_1(\s) [C_2\circ\tau_x](\s)]\er^{(\f)}
=\el C_1(\s)\er^{(\f)}\el C_2(\s)\er^{(\f)}.
\]
\end{lemma}
The proof follows the strategy  
in the appendix of~\cite{adcs}, and is based on
first proving the statement for functions of
the form $C(\s)=\s_A$ using the Griffiths inequality
and then extending to more general functions using the
Stone--Weierstrass theorem.   However, 
there are two  diffculties associated with
this approach:  firstly, the function $C(\s)=\s_A$
is not continuous;  secondly, $\S$ is not compact.
(The locally compact version of the Stone--Weierstrass
theorem is not appropriate since the functions $C$
we want to consider do not `vanish at infinity'.)
Nonetheless, we have the following result.
Let $\cG$ denote the (real) algebra of functions
generated by the monomials of the form $\s_A$
for finite $A\se\KK$.
For tidier notation we drop the superscript $^{(\w)}$
or $^{(\f)}$ in the following result, which holds for
both cases.
\begin{proposition}\label{sw_prop}
Suppose $F:\S\to\RR$ is a bounded and 
measurable function such that
for all $G\in\cG$
\[
\lim_{\|x\|\to\oo}\langle G(\s) [F\circ \tau_x](\s)\rangle
=\langle G(\s)\rangle \langle F(\s)\rangle.
\]
Then for all bounded, 
uniformly continuous $C:\S\to\RR$ we also have 
\[
\lim_{\|x\|\to\oo}\langle C(\s) [F\circ \tau_x](\s)\rangle
=\langle C(\s)\rangle \langle F(\s)\rangle.
\]
\end{proposition}
\begin{proof}
For each $\d>0$ let $\S_\d$ denote the set of functions
$\s\in\S$ which are constant on each interval of the form
$\{x\}\times[k\d,(k+1)\d)$ for $k\in\ZZ$.  Then (by a diagonal
argument or otherwise) $\S_\d$ is compact.  Define a mapping
$\S\to\S_\d$ by letting $\s\mapsto\s_\d$ where 
$\s_\d(x,t)=\s(x,\d\lfloor t/\d\rfloor)$,
and for $F:\S\to\RR$ let $F_\d:\S\to\RR$ be given by
$F_\d(\s)=F(\s_\d)$.  Note that if $G\in\cG$ then $G_\d\in\cG$.
Let $\cC_\d$ denote the set of continuous functions
$\S_\d\to\RR$ and
$\cG_\d$ the set of restrictions of functions in $\cG$
to $\S_\d$.  Then $\cG_\d$
is an algebra of functions in $\cC_\d$, 
and $\cG_\d$ separates
the points of $\S_\d$ (if $\s,\s'\in\S_\d$ differ at the
point $(x,k\d)$ then, by definition, 
$\s(x,k\d)\neq\s'(x,k\d)$).  Thus by the
Stone--Weierstrass theorem $\cC_\d$ is the uniform closure
of $\cG_\d$, meaning that for each 
bounded, uniformly continuous $C:\S\to\RR$
and each $\eps>0$ there
is $G\in\cG$ such that
\[
\sup_{\s\in\S}|G_\d(\s)-C_\d(\s)|=
\sup_{\s\in\S_\d}|G_\d(\s)-C_\d(\s)|
<\eps.
\]
Let $M$ be a uniform upper bound on both $|F|$ and $|C|$.
We have that 
\begin{equation}\label{CMG2}
\begin{split}
| \langle C(\s) [F\circ\tau_x](\s)\rangle &
- \langle G_\d(\s) [F\circ\tau_x] (\s)\rangle |\\
&\leq M \langle |G_\d(\s) -C_\d(\s)|\rangle+M \langle |C(\s)
-C_\d(\s)|\rangle\\
&\leq M\eps + M \langle |C(\s) -C_\d(\s)|\rangle.
\end{split}
\end{equation}
For $\eta>0$ sufficiently small, 
\[\begin{split}
|C(\s) -C_\d(\s)|&=
|C(\s) -C(\s_\d)|\one\{d(\s,\s_\d)<\eta\}\\
&\quad+|C(\s) -C_\d(\s)|\one\{d(\s,\s_\d)\geq\eta\}\\
&\leq \eps + 2M \one\{d(\s,\s_\d)\geq\eta\}.
\end{split}\]
Thus
$\el |C(\s) -C_\d(\s)|\er \leq
\eps+2M \mu(d(\s,\s_\d)\geq\eta)$,
and the last probability converges to 0
as $\d\downarrow 0$ (for example along a sequence
of the form $\d=2^{-m}$).
Hence~\eqref{CMG2} can be made arbitrarily small, uniformly in $x$.
The same bound applies to
$|\el C(\s)\er\el F(\s)\er - \el G_\d(\s)\er\el F(\s)\er |$.
Since $G_\d\in\cG$, the result follows.
\end{proof}

\begin{remark}\label{sw_rk}
Proposition~\ref{sw_prop} holds also if we assume in addition that
$F$, $G$ and $C$ are even functions.  
To prove this in detail one may
pass to the quotient space $\S/\!\!\sim$, where the equivalence relation
$\sim$ consists of all pairs $\{\s,-\s\}$ for $\s\in\S$.
An even function on $\S$ may be identified 
with 
a function on $\S/\!\!\sim$ and this identifies continuous functions with
continuous functions.  (This uses the fact that the mapping
$\s\mapsto-\s$ is an isometry and~\cite[Lemma~3.3.6]{burago}.)
The subspace $\S_\d/\!\!\sim$ is compact and the even functions 
in $\cG$ separate the points of $\S_\d/\!\!\sim$, so we may apply the
Stone--Weierstrass theorem in the same way as in
Proposition~\ref{sw_prop}.  The remaining estimates are the
same. 
\end{remark}

\begin{proof}[Proof of Lemma~\ref{mix_lem}]
We allow ourselves to be rather brief and omit some details.
For the boundary condition $\w$
it suffices to show that for all finite sets $A,B\se\KK$,
\begin{equation}\label{corr_mix_eq}
\el\s_A\s_{B+x}\er^{(\w)}\to \el\s_A\er^{(\w)}\el\s_B\er^{(\w)}
\mbox{ as }\|x\|\to\oo,
\end{equation}
by Proposition~\ref{sw_prop} and linearity.
For the boundary condition $\f$
we need to show that~\eqref{corr_mix_eq}
holds (with $\w$ replaced by $\f$) when $A$ and $B$
are sets of even size, by Remark~\ref{sw_rk}.  

Fix $N_0<N$ and $r_0<r$ 
large enough that $A,B\se K(N_0,r_0)$,
and write $K$ for $K(N,r)$ and $K_0$ for $K(N_0,r_0)$.
We begin by showing that for each bounded,
measurable function $h:K_0\to[0,\oo)$ we have
\begin{multline}\label{exp_mix_eq}
\lim_{\|x\|\to\oo}\Big\langle
\s_A\exp\Big(\sum_{y\in\L_{N_0}}\int_{I_{r_0}}h(y,t)\s(y+x,t)dt\Big)\Big\rangle^{(\w)}
\\=\langle\s_A\rangle^{(\w)}
\Big\langle\exp\Big(\sum_{y\in\L_{N_0}}\int_{I_{r_0}}h(y,t)\s(y,t)dt\Big)\Big\rangle^{(\w)}.
\end{multline}
To go from~\eqref{exp_mix_eq} to~\eqref{corr_mix_eq}
one expands the exponentials as a power series.
Using the fact that~\eqref{exp_mix_eq} holds for arbitrary
$h$ and that correlation functions of the form 
$\el\s(y_1,t_1)\dotsb\s(y_k,t_k)\er^{(\w)}$
are continuous in $t_1,\dotsc,t_k$ one may 
deduce pointwise convergence of the 
form~\eqref{corr_mix_eq} from the corresponding convergence 
of repeated sums and integrals over $y_1,\dotsc,y_k$
and $t_1,\dotsc,t_k$.

We now show~\eqref{exp_mix_eq}.
Write $K(x)=\L_{\|x\|-N_0}\times I_{r_0+\|x\|}$ where
 $x\in\ZZ^d$ is fixed with
$\|x\|$  large enough that $A\se K(x)$.
Let $N,r$ be large enough that $K_0+x\se K$.
Write 
\[
h(\s)=\sum_{y\in\L_{N_0}}\int_{I_{r_0}}h(y,t)\s(y,t)dt
\]
and let $\el\cdot\er_{K;h\circ\tau_x}^\w$ denote the 
wired space--time Ising measure
defined as in~\eqref{strn_eq}--\eqref{stpf_eq} 
but with the additional term 
\begin{equation}\label{h-bar}
h(\tau_x(\s))=
\sum_{y\in\L_{N_0}}\int_{I_{r_0}} h(y,t)\s(y+x,t)dt
\end{equation}
in the exponential.
Using the shorthand 
$\el\cdot\er_K^\w$ for $\el\cdot\er_{N,r}^{\w,\w}$ we have
\begin{equation}\label{KJ_eq}
\big\langle
\s_A\exp\big(h(\tau_x(\s))
\big)\big\rangle_{K}^\w
=\langle\s_A\rangle_{K;h\circ\tau_x}^\w
\big\langle\exp\big( h(\tau_x(\s))
\big)\big\rangle_{K}^\w.
\end{equation}
The Griffiths inequality 
(see~\cite[Lemma~2.2.20]{bjo_phd} for a proof for the present model)
implies that the correlation 
$\el\s_A\er_{K;h\circ\tau_x}^\w$
is increasing when viewed as a function of $h$
(under pointwise ordering of $h$).  Comparison with the
case $h\equiv 0$ gives
\begin{equation*}
\langle\s_A\rangle_{K; h\circ\tau_x}^\w\geq
\langle\s_A\rangle_{K}^\w.
\end{equation*}
If we let $h(y,t)\to\oo$ for all $(y,t)\in K_0$
then $\langle\cdot\rangle_{K; h\circ\tau_x}^\w$ 
converges to a state corresponding
to `wiring' the region $K_0+x$, and we deduce that
\begin{equation*}
\langle\s_A\rangle_{K; h\circ\tau_x}^\w\leq
\langle\s_A\rangle_{K(x)}^\w
\end{equation*}
(cf.\ \cite[Lemma~2.2.22]{bjo_phd}).
Letting $N,r\to\oo$ and applying translation-invariance
we obtain
\begin{equation}\label{Kx_eq}
\begin{split}
\langle\s_A\rangle^{(\w)}
\big\langle\exp\big(h(\s) \big)\big\rangle^{(\w)}
&\leq \big\langle
\s_A\exp\big(h(\tau_x(\s))
\big)\big\rangle^{(\w)}
\\&\leq \langle\s_A\rangle_{K(x)}^\w
\big\langle\exp\big(
h(\s) \big)\big\rangle^{(\w)}
\end{split}
\end{equation}
Letting $\|x\|\to\oo$ we have $K(x)\uparrow\KK$
and hence~\eqref{exp_mix_eq} follows.

For the case of boundary condition 
$\f$ let $J:F(N_0,r_0)\to[-\l,0]$ be measurable
and $q:K_0\to[0,\oo)$ be bounded and measurable.  
Write 
\[
J(\s)=\sum_{yz\in\cE_{N_0}}\int_{I_{r_0}} 
J(yz,t)\s(y,t)\s(z,t)dt,
\]
and (recalling the process $D$ of 
discontinuities of $\s$)
\[
q(\s)=\sum_{(y,t)\in D\cap K_0} q(y,t).
\]
Note that $q(\s)$ is a function of $D$ only, and we
may therefore write  $q(D)$ in place of $q(\s)$.  With this notation
we have
$q(\tau_x(\s))= q(\tau_{-x}(D))$.
Let $\el\cdot\er_{K; (J,q)\circ\tau_x}^\f$ denote the measure 
defined as in~\eqref{strn_eq}--\eqref{stpf_eq} 
but with the additional term 
$J(\tau_x(\s))+q(\tau_x(\s))$
in the exponential.  We have that
\begin{multline}\label{Jq_eq}
\big\langle
\s_A\exp\big(J(\tau_x(\s))+q(\tau_x(\s))
\big)\big\rangle_{K}^\f\\
=\langle\s_A\rangle_{K;(J,q)\circ\tau_x}^\f
\big\langle\exp\big( J(\tau_x(\s))+q(\tau_x(\s))
\big)\big\rangle_{K}^\f.
\end{multline}
By standard properties of Poisson processes, 
$\el\cdot\er_{K; (J,q)\circ\tau_x}^\f$ 
may alternatively be obtained by first modifying the intensity
of $D$ under the a-priori measure $E_{N,r}$ from the constant
intensity $\d$ to the variable intensity 
\[
\d(y,t)=\d\exp(q(y-x,t)\one\{(y,t)\in K_0+x\}),
\]
and then having only the additional term $J(\tau_x(\s))$
in the exponential.  
Thus the correlation $\langle\s_A\rangle_{K;(J,q)\circ\tau_x}^\f$
is increasing in $J$ and decreasing in $q$,
and comparison with the cases $J\equiv-\l$,
$q\equiv\oo$, respectively $J\equiv0$, $q\equiv0$,
gives 
\[
\langle\s_A\rangle_{K(x)}^\f\leq
\langle\s_A\rangle_{K;(J,q)\circ\tau_x}^\f\leq
\langle\s_A\rangle_{K}^\f.
\]
Similarly to~\eqref{exp_mix_eq} we deduce that 
\[
\lim_{\|x\|\to\oo}\big\langle
\s_A\exp\big(J(\tau_x(\s))+q(\tau_x(\s))\big)
\big\rangle^{(\f)}=\langle\s_A\rangle^{(\f)}
\big\langle
\exp\big(J(\s)+q(\s)\big)
\big\rangle^{(\f)}.
\]
Expanding the exponential
$\exp\big(J(\s)\big)$ as for~\eqref{corr_mix_eq}
we deduce that for all finite $B\se\FF$,
\begin{equation}\label{q-lim}
\el\s_A\s_{B+x}\exp\big(q(\tau_x(\s))\big)\er^{(\f)}
\to \el\s_A\er^{(\f)}\el\s_B\exp\big(q(\s)\big)\er^{(\f)}
\mbox{ as }\|x\|\to\oo.
\end{equation}
Let $x_1,\dotsc,x_n\in\ZZ^d$ and let
$s_1<t_1,\dotsc,s_n<t_n$ be real numbers
such that all points of the form $(x_j,s_j)$
or $(x_j,t_j)$ are distinct, and let 
$B'$ be the set of these points.  (Thus $B'\cap (\{x\}\times\RR)$
has even size for all  $x\in\ZZ^d$.)
One may deduce from~\eqref{q-lim} that for any such set $B'\se\KK$ we
have that  
\begin{equation}\label{q-lim-2}
\el\s_A\s_{B+x}\s_{B'+x}\er^{(\f)}
\to \el\s_A\er^{(\f)}\el\s_B\s_{B'}\er^{(\f)}
\mbox{ as }\|x\|\to\oo.
\end{equation}
This proves the claim of the lemma for the boundary condition $\f$
since for any set $B\se\KK$ of even size one may write
$\s_B=\s_{B'}\s_{B''}$ for some set $B'$ as above, and some
finite $B''\se\FF$.  

One way to see that~\eqref{q-lim} implies~\eqref{q-lim-2} is as
follows (we give only a sketch).  
One may see~\eqref{q-lim} as a result about convergence of
the Laplace functionals of the point processes
$\tau_{x}^{-1}(D)\cap K_0$
with certain `skewed' distributions.  
Using Theorem~11.1.VI and Proposition~11.1.VII
of~\cite{dvj} it follows that 
\[
\el\s_A\s_{B+x}\one\{\tau^{-1}_{x}(D)\in\cC\}\er^{(\f)}
\to \el\s_A\er^{(\f)}\el\s_B\one\{D\in\cC\}\er^{(\f)}
\]
for each `stochastic continuity set'  $\cC$.  These sets include the
events 
\[
\cC_j=\{|D\cap (\{x_j\}\times[s_j,t_j))|\mbox{ is even}\},
\]
and using the identity
$\s(x_j,s_j)\s(x_j,t_j)=2\one_{\cC_j}-1$
one may write $\s_{B'}$ as a linear combination of terms of the form 
$\one\{\bigcap_{j\in J}\cC_j\}$ for $J\se\{1,\dotsc,n\}$.
We deduce~\eqref{q-lim-2} by linearity.
\end{proof}

\subsection*{Acknowledgement}

The author thanks Geoffrey Grimmett for drawing his attention
to the article~\cite{adcs}, and the two anonymous referees for their
helpful comments, corrections and suggestions.

\end{document}